%% file: main.tex
\documentclass[aps,twocolumn,amsmath,amssymb,nofootinbib]{revtex4-2}

\usepackage{graphicx}
\usepackage{amsfonts}
\usepackage{amsthm}
\usepackage{xcolor}
\usepackage[normalem]{ulem}

\usepackage{booktabs}
\usepackage{algpseudocode}
\usepackage{tikz}
\usetikzlibrary{arrows.meta,positioning,fit,backgrounds,calc}
\usepackage{hyperref}
\hypersetup{colorlinks=true,linkcolor=blue,citecolor=blue,urlcolor=blue,filecolor=blue}
\usepackage[capitalize]{cleveref}

\providecommand{\mathbbm}[1]{\mathbf{#1}}

\theoremstyle{plain}
\newtheorem{theorem}{Theorem}

\theoremstyle{remark}
\newtheorem{remark}{Remark}

\begin{document}

\title{\texttt{DiffVortex}: A quantum-inspired fast solver for incompressible flows}

\author{Yu Jun Loo}
\email{looyujun@umich.edu}
\author{Silas Alben}
\author{Shravan Veerapaneni}
\affiliation{Department of Mathematics, University of Michigan, Ann Arbor, MI 48109, USA}

\begin{abstract}
Matrix-product-state (MPS) and quantized tensor-train (QTT) representations have recently emerged as a promising framework for computational fluid dynamics, with run time and memory scaling logarithmically in the mesh size when the bond dimension stays bounded. Existing MPS-based solvers are Eulerian: they encode the velocity field directly as an MPS and enforce immersed-object boundary conditions through static masks or body-fitted coordinates. Generalization to flexible or moving bodies is not straightforward in this framework: a deforming body necessitates a fresh MPS mask at every time step, and high-Reynolds-number flows develop filamentary vorticity structures that are expected to resist low-rank compression of the field itself. We propose a complementary route. Instead of compressing the flow state, we compress the Biot--Savart kernel that links vorticity to velocity, and we evolve the flow with a Lagrangian vortex-particle method, representing each immersed body by a bound vortex sheet on Lagrangian surface points. Our main theoretical result, a lifting theorem, shows that bounded QTT rank $r$ of the kernel implies bounded matrix-product-operator (MPO) rank of the associated convolution operator, independently of grid resolution; for the interleaved ordering used here, the explicit bound is $2^d r$ in $d$ dimensions and is sharp. Combined with quantics tensor cross interpolation, this turns the moment-to-local stage of a single-level fast multipole method into a precomputed tensor-network contraction with quasi-linear complexity, while preserving full Lagrangian flexibility for the embedded bodies. We implement the algorithm (\texttt{DiffVortex}) in PyTorch, so that the entire simulation pipeline is differentiable end-to-end; it runs on GPUs and reaches $6.7\times10^{7}$ particles on a single device. We validate the solver against reference vortex-method calculations for an impulsively started cylinder and against exact and reference solutions for inviscid elliptical vortices, and demonstrate it on two-dimensional flows past rigidly moving, multiple, and deforming bodies, including gradients propagated through the full time integration.
\end{abstract}

\maketitle

\section{Introduction}\label{sec:intro}

The discovery that the velocity and pressure fields of incompressible Navier--Stokes flows are often well approximated by matrix-product states (MPSs), the tensor networks that underlie the density-matrix renormalization group (DMRG) for one-dimensional quantum many-body systems, has opened a fast-growing research direction at the interface of tensor networks and computational fluid dynamics. Gourianov \emph{et al.}~\cite{gourianov2022quantum} first showed that the velocity field of two-dimensional decaying turbulence admits compact MPS representations, with the energy-cascade mechanism playing a role analogous to the local interactions that make one-dimensional quantum states compressible. 
Peddinti \emph{et al.}~\cite{peddinti2024quantum} extended this observation into a complete incompressible Navier--Stokes solver. A grid with $2^L$ cells per dimension is reshaped into $dL$ binary modes, where $d$ is the spatial dimension, and the velocity field is encoded as an MPS over these modes (a quantized tensor train, or QTT). Time stepping combines finite differences with Chorin projection, the pressure Poisson equation being solved with DMRG. Memory and per-step run time then scale as $\mathcal{O}(L\,\mathrm{poly}(\chi))$ rather than $\mathcal{O}(2^{dL})$, where $\chi$ is the MPS bond dimension, provided $\chi$ remains bounded as the grid is refined, such as when fine-scale structure is weakly correlated across scales. H\"olscher \emph{et al.}~\cite{holscher2025quantum} accelerated the MPS algorithm of Gourianov \emph{et al.} on GPUs using cuQuantum, simulating two-dimensional decaying jets and decaying turbulence at Reynolds numbers up to $10^7$, and van H\"ulst \emph{et al.}~\cite{vanhulst2025curvilinear} developed a tensor-network fractional-step method in body-fitted curvilinear coordinates, demonstrated on stationary and rotating cylinders, with aerodynamic and aeroacoustic applications discussed in~\cite{vanhulst2025aerodynamic}. The framework has since been extended to three-dimensional turbulence and synthetic-turbulence generation~\cite{pisoni2025compression}, compressible Euler flows~\cite{peddinti2025compressible}, reacting shear flows~\cite{pinkston2025reacting}, lattice Boltzmann discretizations~\cite{gross2026tensor}, finite-element formulations~\cite{kornev2024tetrafem}, pseudospectral methods~\cite{gidi2025pseudospectral}, and probability distribution functions of turbulence~\cite{gourianov2025tensor}; Amaral \emph{et al.}~\cite{amaral2025review} provide a review of this rapidly developing field.

Immersed bodies pose a particular challenge for this Eulerian formulation. The state-of-the-art treatment, introduced in~\cite{peddinti2024quantum}, is the \emph{approximate MPS mask}: a smooth surrogate $m(\boldsymbol{x})$ of the body indicator function $\theta(\boldsymbol{x})$ is compressed by TT-cross interpolation~\cite{oseledets2010tt,nunez2025learning} and applied element-wise to the velocity MPS at every time step. The construction is elegant for static rigid bodies and has been extended to curvilinear coordinates aligned with simple rotating bodies~\cite{vanhulst2025curvilinear}; the lattice-Boltzmann variant of~\cite{gross2026tensor}, which enforces no-slip conditions through MPS-compressed bounce-back, is likewise tied to a static geometry. To our knowledge, no MPS/QTT-based CFD solver in the published literature handles deforming bodies, or moving bodies beyond these special cases. Two difficulties arise in extending Eulerian MPS solvers in this direction. First, for a deforming body the mask must be rebuilt at every step. Each rebuild is a TT-cross construction whose cost grows with the bond dimension of the mask, and that bond dimension grows quickly for thin or sharp geometries~\cite{peddinti2024quantum}. Rigid translations by an integer number of grid cells can instead be applied exactly with MPO shift operators; no comparable operator exists for a general deformation.

Second, compressibility of the velocity field itself is not guaranteed. Low-rank MPS descriptions have been demonstrated mainly for two-dimensional decaying turbulence and moderate-Reynolds-number flows~\cite{gourianov2022quantum,peddinti2024quantum}. At $\mathrm{Re}=10^7$, H\"olscher \emph{et al.}~\cite{holscher2025quantum} find that the bond dimension required at a fixed truncation error saturates, but that this truncation no longer reproduces the energy spectrum, and that Kelvin--Helmholtz roll-up of shear layers raises the required bond dimension. Flows past moving bodies are dominated by such structures: vortex sheets shed from the body roll up into thin filaments that are nearly discontinuous at the grid scale, and are therefore expected to inflate the bond dimension needed to represent the velocity field.

We propose a quantum-inspired solver that addresses both difficulties by compressing the operator rather than the flow state: instead of encoding the velocity field as an MPS, we compress the (regularized) Biot--Savart kernel that maps vorticity to velocity. Because the kernel depends only on relative position, it is fixed for the entire simulation: its matrix-product-operator (MPO) representation is precomputed once in an offline stage, and reused at every time step regardless of how the body and the vortex particles evolve. The flow itself is represented by Lagrangian vortex particles, which naturally trace the vorticity field and accommodate filamentary structures without tensor-compression error; each immersed body is represented by a bound vortex sheet on Lagrangian surface points, which releases vorticity into the particle field. This approach complements the framework of~\cite{peddinti2024quantum} in two respects: it eliminates the need to re-mask the velocity at each time step, and it removes the requirement that the velocity field itself admit a low-rank MPS representation.

The central technical observation underlying this strategy is a \emph{lifting theorem}. The use of QTT representations for translation-invariant operators goes back to Khoromskij~\cite{khoromskij2011quantics} and Kazeev \emph{et al.}~\cite{kazeev2013multilevel}, who studied multilevel Toeplitz matrices and high-dimensional convolution; more recently, Chen \emph{et al.}~\cite{chen2023quantum} showed that the discrete Fourier transform itself admits a low-rank MPO representation in quantics coordinates. In particular, for the serial (also called fused) ordering of the binary coordinates, Kazeev \emph{et al.} showed that a generating vector of MPS (QTT) rank $p$ yields a general multilevel Toeplitz matrix of MPO rank at most $2p$. Their generating vector encodes both positive and negative coordinate differences. Here we establish the corresponding result for the interleaved ordering used in our algorithm: if a $d$-dimensional kernel has MPS rank $r$ when sampled on a uniform grid in interleaved binary coordinates, then the associated one-sided (triangular) convolution operator admits an MPO of rank at most $2^d r$. For the one-sided operators considered here, the corresponding serial and interleaved rank-inflation constants are sharp, as demonstrated by separable exponential kernels. Thus, under either of the two natural bit orderings, grid-independent tensor rank of the kernel implies grid-independent tensor rank of the convolution operator, although the kernel ranks themselves depend on the ordering.

This lifting theorem suggests that tensor compression is useful not only for storage, but also for accelerating the long-range interaction itself. Corona \emph{et al.}~\cite{corona2017tensor} used QTT to compress and invert discretized boundary-integral operators, exploiting their translation invariance on uniform grids, and, building on TT-cross interpolation~\cite{oseledets2010tt}, recent advances in quantics tensor cross interpolation (QTCI)~\cite{nunez2025learning,ghahremani2024cross} have made black-box construction of such tensor representations practical for a wide class of smooth kernels. We use these tools to construct a compact MPO representation of the moment-to-local (M2L) stage of a single-level fast multipole method (FMM)~\cite{greengard1987fast} directly from kernel evaluations. When the M2L operator has bounded MPO rank, the long-range step becomes a tensor-network contraction whose cost scales linearly in the number of boxes, up to logarithmic factors, removing the cost quadratic in the number of boxes that otherwise restricts single-level FMM to $\mathcal{O}(N^{4/3})$ complexity~\cite{martinsson2019fast}. Tensor compression has previously been used in computational electromagnetics to reduce the storage cost of FMM translation operators and multipole coefficients~\cite{qian2020compression,qian2019tensor,yucel2017compression}; here the compressed representation instead plays a direct algorithmic role, reducing both memory and application cost. Although ordinary TT compression could in principle provide the same acceleration, this would require low separation rank across the physical coordinates, whereas QTT exposes low-rank structure across binary scales and can therefore remain efficient even when such physical-coordinate separability is absent. Like ordinary TT compression, the rank of the serial QTT ordering depends on the separation rank across physical coordinates, hence we use the interleaved ordering in this work.

We package the resulting algorithm as \texttt{DiffVortex}, implemented in PyTorch~\cite{paszke2019pytorch} so that the entire simulation is differentiable end-to-end (\Cref{sec:method}); gradients of a downstream objective with respect to particle, kernel, or body-shape parameters then follow directly, without hand-derived adjoint equations. Immersed bodies enter only through the surface points, the bound vortex sheet and the wall vorticity flux, so prescribed moving and deforming bodies are handled by the same machinery as rigid ones; freely moving bodies require only coupling to a body model (\Cref{sec:discussion}). No compression is applied to the instantaneous flow state: all tensor-compression error is confined to the static, precomputed moment-to-local tensor and can be reduced by tightening the QTCI tolerance, at the cost of a larger rank.

The contributions of this paper are threefold. (i)~An interleaved-ordering lifting theorem: a one-sided convolution operator generated by a kernel of MPS rank $r$ admits an MPO of rank at most $2^d r$, complementing the serial-ordering construction of Kazeev \emph{et al.}~\cite{kazeev2013multilevel}; the rank-inflation constants $2$ and $2^d$ are sharp for the one-sided operators considered here. (ii)~\texttt{DiffVortex}, a differentiable single-level fast-multipole vortex-particle solver in which the moment-to-local operator is precomputed once as an MPO and reused unchanged at every time step; combined with the classical bound-vortex-sheet treatment of immersed bodies~\cite{koumoutsakos1995high}, it handles moving and deforming bodies, requiring neither re-masking nor a low-rank velocity field. (iii)~Numerical evidence: MPO ranks of the moment-to-local operator that do not grow with grid size, scaling to $6.7\times10^{7}$ particles on a single GPU, validation against reference vortex-method and exact solutions, demonstrations with moving, multiple and deforming bodies, and gradients propagated through the full time integration.

The paper is organized as follows. \Cref{sec:formulation} states the vorticity formulation and the treatment of immersed bodies; \Cref{sec:mpo} states the lifting theorem; \Cref{sec:method} describes the \texttt{DiffVortex} algorithm and its complexity; \Cref{sec:examples} presents numerical results; and \Cref{sec:discussion} discusses extensions. The proof of the lifting theorem, the comparison of the two orderings, and the moment-to-local coefficients are given in the appendices.

\section{Formulation}\label{sec:formulation}

We consider the two-dimensional incompressible Navier--Stokes equations in the
fluid domain
\[
    \Omega_f(t)=\Omega\setminus\mathcal{B}(t),
\]
where $\mathcal{B}(t)$ denotes the region formed by the union of one or more immersed bodies:
\begin{align}
\partial_t \boldsymbol{u}
+(\boldsymbol{u}\cdot\nabla)\boldsymbol{u}
&=
-\frac{1}{\varrho}\nabla p
+\nu\nabla^2\boldsymbol{u},
\label{eq:ns_mom}\\
\nabla\cdot\boldsymbol{u}&=0.
\label{eq:ns_inc}
\end{align}
where $\varrho$ is the fluid density, $p$ the pressure, and $\nu$ the kinematic viscosity.
On the body surface we prescribe the material velocity
$\boldsymbol{u}_{\partial\mathcal{B}}$ and impose no-slip,
\begin{equation}
    \boldsymbol{u}
    =
    \boldsymbol{u}_{\partial\mathcal{B}}
    \qquad\text{on }\partial\mathcal{B}(t).
\end{equation}

Introducing the scalar vorticity
\[
    \omega = \partial_x u_y-\partial_y u_x,
\]
and taking the curl of \eqref{eq:ns_mom} gives
\begin{equation}
    \partial_t\omega
    +\boldsymbol{u}\cdot\nabla\omega
    =
    \nu\nabla^2\omega
    \qquad \text{in }\Omega_f(t).
\label{eq:vorticity}
\end{equation}
The velocity induced by the fluid vorticity is given by the Biot--Savart law (the contribution of the bound vortex sheet on $\partial\mathcal{B}$ is added in \Cref{sec:embedding}),
\begin{align}
\boldsymbol{u}(\boldsymbol{x},t) &= \boldsymbol{u}_\infty(\boldsymbol{x},t)
\nonumber\\
&\quad+ \int_{\Omega_f(t)}\boldsymbol{K}(\boldsymbol{x}- \boldsymbol{x}')\,
\omega(\boldsymbol{x}',t)\,dA_{\boldsymbol{x}'},
\label{eq:biot_savart_unreg}
\end{align}
where, with $\boldsymbol{r}^{\perp}=(-r_2,r_1)$, the Biot--Savart kernel is
\begin{equation}
\boldsymbol{K}(\boldsymbol{r})
=
\frac{1}{2\pi}
\frac{\boldsymbol{r}^{\perp}}{|\boldsymbol{r}|^2}.
\label{eq:bs_kernel}
\end{equation}
and $\boldsymbol{u}_\infty$ is the imposed background (free-stream) velocity. \Cref{eq:biot_savart_unreg} implicitly enforces incompressibility. In the vorticity--velocity formulation, the tangential no-slip condition is
not imposed directly on $\omega$. Instead, following Lighthill~\cite{lighthill1963boundary} (see also~\cite{wu1998boundary}), the vorticity generated at a
viscous solid boundary enters through a Neumann-type vorticity-flux
condition on $\partial\mathcal{B}(t)$
\cite{koumoutsakos1994boundary,koumoutsakos1995high,
ploumhans2000vortex} to cancel the slip velocity. Numerically, this flux is incorporated into the
viscous update of the vortex particles so that the reconstructed fluid velocity
matches the prescribed wall motion. We use the classical vortex-method
treatment of this boundary process~\cite{koumoutsakos1994boundary,koumoutsakos1995high,
ploumhans2000vortex} and focus here on the acceleration of
the resulting Biot--Savart interactions.

The vorticity is represented by a collection of Lagrangian point vortices,

\begin{equation}
     \omega(\boldsymbol{x},t) \approx \sum_{j=1}^{N}\Gamma_j(t)\, \delta_{\boldsymbol{x}_j(t)}(\boldsymbol{x}).
\end{equation}

Here $\boldsymbol{x}_j$ and $\Gamma_j$ denote the particle positions and
circulations, and $\delta_{\boldsymbol{x}_j(t)}$ is the Dirac mass at the particle position $\boldsymbol{x}_j(t)$. Whenever a mesh-based approximation of the vorticity is required, these particle strengths are interpolated onto a Cartesian grid with mesh spacing $\Delta x=\mathcal{O}(\delta)$, with $\delta$ the blob radius of \eqref{eq:biot_savart_kernel}, using the compactly supported interpolating $M_4'$ kernel~\cite{eldredge2007numerical}. The particles are advected with the local fluid velocity,
while viscosity and wall-generated vorticity modify their circulations
according to the vortex-particle discretization described in
\Cref{sec:method}. The diffusion step is applied with particle strength exchange (PSE)~\cite{koumoutsakos1995high} with kernel width $\sigma = \sqrt{2\nu\Delta t}$, and the mesh spacing is set to $\Delta x = \tfrac{2}{3}\sigma$ with blob radius $\delta = \Delta x$, so that the overlap condition holds; these choices fix the diffusion number at $\nu\Delta t/\Delta x^{2}=9/8$. The particles are periodically remeshed onto this lattice to control Lagrangian distortion~\cite{koumoutsakos2008flow}, so $\Delta x$ is also the particle spacing. In the inviscid calculations of \Cref{sec:results-physical} there is no PSE step and $\sigma$ plays no role; there $\Delta x$ is fixed by the number of particles laid over the vortex support and $\delta = \tfrac{3}{2}\Delta x$.

The velocity associated with the particle distribution is recovered from
the regularized Biot--Savart law,
\begin{equation}
    \boldsymbol{u}(\boldsymbol{x},t)
    =
    \boldsymbol{u}_\infty(\boldsymbol{x},t)
    +
    \sum_{j=1}^{N}
    \boldsymbol{K}_\delta
    \!\left(
        \boldsymbol{x}-\boldsymbol{x}_j(t)
    \right)
    \Gamma_j(t),
\label{eq:biot_savart}
\end{equation}
where, for the vortex-blob regularization used here~\cite{krasny1986desingularization},
\begin{equation}
    \boldsymbol{K}_\delta(\boldsymbol{r})
    =
    \frac{1}{2\pi}
    \frac{\boldsymbol{r}^{\perp}}
    {|\boldsymbol{r}|^2+\delta^2}.
\label{eq:biot_savart_kernel}
\end{equation}
This is a regularized point-vortex quadrature of \eqref{eq:biot_savart_unreg}; the bound-sheet contribution is treated separately in \Cref{sec:embedding}.

Like the singular kernel, $\boldsymbol{K}_\delta$ depends only on the relative displacement $\boldsymbol{r}$, so the velocity evaluation is a translation-invariant $N$-body sum. This is the property that all fast summation methods exploit. In a fast multipole method on a uniform box grid, for example, the moment-to-local translations depend only on relative box positions and are precomputed once~\cite{greengard1987fast,martinsson2019fast}; kernel-independent variants further compress them by SVD or FFT~\cite{ying2004kernel,fong2009black}, and adaptive treecodes apply the same idea to regularized vortex particles~\cite{lindsay2001particle}. Our departure is in how this fixed operator is compressed. Rather than treating each translation separately, we compress the full moment-to-local tensor, over all box offsets and expansion indices, as a single MPO in quantized coordinates (\Cref{sec:mpo,sec:method}).

\subsection{Embedding flexible bodies}\label{sec:embedding}

A flexible body is represented in two dimensions by an ordered set of
surface points
$\{\boldsymbol{y}_k(t)\}_{k=1}^{N_\partial}$; when several bodies are present, $N_\partial$ denotes the total number of surface points over all of them.
For an actively deforming body the configuration
$\boldsymbol{y}_k(t)$ is prescribed, while for a passive body it evolves
under the appropriate structural equations. The body enters the vortex
formulation through a bound vortex sheet supported on
$\partial\mathcal{B}(t)$, together with the viscous production and release
of vorticity from the surface.

At each time step, the instantaneous bound-sheet strength is determined
from the no-penetration (impermeability) condition
\begin{equation}
\boldsymbol{n}(\boldsymbol{y}_k,t)\cdot
\left[
\boldsymbol{u}_{\omega}(\boldsymbol{y}_k,t)
+
\boldsymbol{u}_{\gamma}(\boldsymbol{y}_k,t)
-
\boldsymbol{u}_{\partial\mathcal{B}}(\boldsymbol{y}_k,t)
\right]
=0,
\label{eq:no_penetration}
\end{equation}
where
\begin{equation}
\boldsymbol{u}_{\omega}(\boldsymbol{y}_k,t)
=
\sum_{j=1}^{N}
\boldsymbol{K}_{\delta}
\bigl(\boldsymbol{y}_k-\boldsymbol{x}_j\bigr)\Gamma_j(t)
\end{equation}
is the velocity induced by the free vortex particles, while
\begin{equation}
\boldsymbol{u}_{\gamma}(\boldsymbol{y}_k,t)
=
\operatorname{p.v.}
\int_{\partial\mathcal{B}(t)}
\boldsymbol{K}
\bigl(\boldsymbol{y}_k-\boldsymbol{y}'\bigr)
\gamma(\boldsymbol{y}',t)\,ds'
\end{equation}
is the velocity induced by the bound vortex sheet, evaluated with the
\emph{unregularized} Biot--Savart kernel~\eqref{eq:bs_kernel}. Discretizing the sheet at the surface points gives the unknown strengths $\gamma_k(t)\approx\gamma(\boldsymbol{y}_k,t)$, $k=1,\dots,N_\partial$, from which the circulation released by panel $k$ is $\gamma_k\,\Delta s_k$ with $\Delta s_k$ the panel arc length.
The resulting boundary-integral equation is supplemented by the
appropriate circulation constraint, so that the bound-sheet solution is
consistent with Kelvin's circulation theorem
\cite{koumoutsakos1994boundary,koumoutsakos1995high,
ploumhans2000vortex, eldredge2007numerical}.

The no-penetration condition determines the instantaneous bound vortex
sheet, but does not by itself enforce no-slip. The remaining tangential
velocity mismatch on the fluid side of the surface represents the
vorticity that must be generated by viscous stresses at the wall.
Following classical viscous vortex methods
\cite{koumoutsakos1994boundary,koumoutsakos1995high,
ploumhans2000vortex}, the bound vortex sheet is released from the surface into the
fluid as a wall vorticity flux and incorporated into the free vortex-particle
population. In the present implementation the sheet is re-solved against the
current particle field at every step and shed in its entirety, distributed
over a diffusion length $\sqrt{2\nu\Delta t}$ normal to the surface. In this way the no-slip condition is recovered through vorticity production, rather than imposed directly as a boundary condition on the
vorticity--velocity formulation.

The two velocity evaluations therefore play distinct roles. Interactions
involving the bound vortex sheet are evaluated with the singular boundary
kernel, whereas interactions among free particles use the regularized
kernel $\boldsymbol{K}_{\delta}$. The tensor-network acceleration
developed below applies to the latter long-range particle interactions;
the body geometry itself remains fully Lagrangian and may translate,
rotate, or deform without reconstruction of an Eulerian mask.

\section{The lifting theorem}\label{sec:mpo}

\subsection{Quantized coordinates and MPO contractions}\label{sec:bitinterleaved}
Following~\cite{khoromskij2011quantics,oseledets2010tt,corona2017tensor,lindsey2023multiscale}, we write a one-dimensional index $i\in\{0,\dots,2^L-1\}$ in binary as

\[
i=\sum_{k=0}^{L-1}i^{(k)}2^k,
\]

and interpret a vector $c$ of length $2^L$ as an $L$-fold tensor
$\widetilde c_{i^{(0)}\cdots i^{(L-1)}}=c_i$. A $2^L\times 2^L$ matrix $A$ is correspondingly reshaped as
\begin{equation}
\widetilde A_{i^{(0)}j^{(0)}\cdots i^{(L-1)}j^{(L-1)}}=A_{ij}.
\end{equation}
If this tensor admits an MPO representation with maximal bond dimension (rank) $r$, it can be applied to a dense vector in $\mathcal{O}(r^2L\,2^L)$ operations by contracting one bit level at a time~\cite{corona2017tensor}.

In $d$ dimensions, there are two natural orderings of the quantized coordinates. The \emph{serial} ordering groups all bits belonging to the same physical coordinate, whereas the \emph{interleaved} ordering groups the $d$ coordinate bits belonging to the same dyadic level into one core, so that the cores run over levels $k=0,\dots,L-1$, each carrying $(i_1^{(k)},\dots,i_d^{(k)})$.

\subsection{Statement}\label{sec:lifting-statement}
Let a scalar function $f:[0,1]^d\to\mathbb{R}$ be sampled on the uniform grid $x_i=i/(2^L-1)$, $i=0,\dots,2^L-1$, in each coordinate, giving the \emph{sampled kernel} $f_{\boldsymbol{i}}=f(x_{i_1},\dots,x_{i_d})$ with $2^{Ld}$ entries. We denote the maximal MPS (QTT) ranks of $f_{\boldsymbol{i}}$ in the serial and interleaved orderings by $r_{\mathrm S}$ and $r_{\mathrm I}$, respectively. The associated one-sided (triangular) discrete convolution operator is
\begin{equation}
(T_f)_{\boldsymbol{i},\boldsymbol{j}}= f(x_{i_1}-x_{j_1},\dots,x_{i_d}-x_{j_d})\,
\mathbbm{1}_{i_1\geq j_1,\dots,i_d\geq j_d}.
\label{eq:discrete-convolution}
\end{equation}

\begin{theorem}[Lifting]\label{thm:lifting}
Let $r_{\mathrm S}$ and $r_{\mathrm I}$ be the exact maximal MPS ranks of the sampled kernel $f_{\boldsymbol{i}}$ in the serial and interleaved orderings, respectively. Then the convolution operator $T_f$ in \eqref{eq:discrete-convolution} admits MPO representations, in the same two orderings, whose maximal ranks $R_{\mathrm S}$ and $R_{\mathrm I}$ satisfy
\begin{equation}
R_{\mathrm S}\leq 2r_{\mathrm S},
\qquad
R_{\mathrm I}\leq 2^d r_{\mathrm I}.
\label{eq:serial-interleaved-lifting}
\end{equation}
For $L\geq2$, both constants are sharp.
\end{theorem}

These estimates follow from the observation that bit-by-bit subtraction generates one borrow bit per coordinate, which can be threaded through the bond dimension as a finite-state machine. Carrying the arithmetic carry on the bond index is the same device used for QTT shift and difference operators~\cite{kazeev2012laplace} and for affine coordinate transformations of quantics variables~\cite{rohshap2025twoparticle}. The factor $2^d$ counts the possible states of the joint borrow. Sharpness follows from the separable exponential

\[
f(\boldsymbol{x})=\prod_{\mu=1}^d a_\mu^{x_\mu},\qquad a_\mu>0,
\]

which has MPS rank one in either ordering but generates convolution operators of exact MPO rank $2$ in the serial ordering and $2^d$ in the interleaved ordering. A complete proof is given in \Cref{app:proof}.

The theorem also transfers to approximate ranks in the max norm. Every entry of $T_f$ is either zero or an entry of the sampled kernel, so if $\tilde f_{\boldsymbol{i}}$ has rank $r$ and $\max_{\boldsymbol{i}}|f_{\boldsymbol{i}}-\tilde f_{\boldsymbol{i}}|\le\varepsilon$, then $T_{\tilde f}$ has rank at most $2r$ (serial) or $2^d r$ (interleaved) and satisfies $\max_{\boldsymbol{i},\boldsymbol{j}}|(T_f-T_{\tilde f})_{\boldsymbol{i},\boldsymbol{j}}|\le\varepsilon$.

\begin{remark}
The two estimates are conditional on the rank of the kernel in their respective orderings and therefore do not imply that either ordering is uniformly preferable. Indeed, $r_{\mathrm S}$ and $r_{\mathrm I}$ arise from different tensor unfoldings and need not have the same dependence on either $L$ or $d$. Consequently, the explicit factors $2$ and $2^d$ do not capture the full dimension dependence of the resulting MPO ranks. Furthermore, \Cref{thm:lifting} concerns exact algebraic tensor ranks. In practice, the kernel and operator are approximated at a prescribed tolerance, and their numerical ranks may be substantially smaller than the exact bounds. Beyond the max-norm transfer noted above, neither estimate in \eqref{eq:serial-interleaved-lifting} controls this singular-value truncation; the numerical behavior of the two orderings is compared in \Cref{app:ordering-comparison}, where we find that the interleaved ordering has lower (for $\sqrt{b}\ge32$) and more stable maximum ranks for the kernel considered in this study, despite having a larger explicit factor.
\end{remark}

\begin{remark}
For smooth kernels, the approximate MPS rank of the sampled kernel is controlled by its regularity rather than directly by the mesh, rigorously so in one dimension~\cite{lindsey2023multiscale}. The lifting theorem shows that, for either ordering, grid-independent kernel rank (exact, or approximate in the max norm) implies grid-independent operator rank. Numerically, the corresponding MPO ranks remain nearly bounded for the real-analytic kernel considered here ($d=2$; see \Cref{app:ordering-comparison} and~\cite{nunez2025learning}).
\end{remark}

\begin{remark}
Combined with rank bounds for smooth kernels, the lifting theorem explains why the convolution operators they induce can be compressed in quantized coordinates even when the kernels are not separable in physical coordinates. It complements the empirical observations of~\cite{gourianov2022quantum,peddinti2024quantum,holscher2025quantum} that moderate-Reynolds-number velocity \emph{fields} admit compact MPS representations.
\end{remark}

\begin{remark}
The preceding result extends to two-sided convolution operators.
For each subset $S\subseteq\{1,\ldots,d\}$, define the reflection
\[
\alpha_S[f](x_1,\ldots,x_d)
=
f(\sigma_1x_1,\ldots,\sigma_dx_d),
\]
where
\[
\sigma_k=
\begin{cases}
-1, & k\in S,\\
 1, & k\notin S.
\end{cases}
\]
Suppose that the restriction of $\alpha_S[f]$ to the grid on
$[0,1]^d$ admits a bit-interleaved MPS representation of rank
at most $r$ for every $S$. Then the two-sided convolution operator
\[
\widetilde F_{i_1,j_1,\ldots,i_d,j_d}
=
f(x_1-y_1,\ldots,x_d-y_d)
\]
admits an MPO representation with rank at most
\[
5^d r-r+1.
\]
To see this, let
\[
P_k=\mathbbm{1}_{\{i_k\ge j_k\}},\qquad
N_k=\mathbbm{1}_{\{i_k\le j_k\}},\qquad
D_k=\mathbbm{1}_{\{i_k=j_k\}}.
\]
Since $P_k+N_k-D_k=1$, we have the exact identity
\[
\begin{aligned}
\widetilde F_{\mathbf i,\mathbf j}
&=f(\mathbf x-\mathbf y)
\prod_{k=1}^d(P_k+N_k-D_k)\\
&=
\sum_{K\subseteq\{1,\ldots,d\}}
(-1)^{|K|}\\
&\qquad\times
\left(\prod_{k\in K}D_k\right)
f(\mathbf x-\mathbf y)
\prod_{\ell\notin K}(P_\ell+N_\ell).
\end{aligned}
\]
Expanding the remaining product (i.e.\ $\prod_{\ell\notin K}(P_\ell+N_\ell)$) yields
$2^{d-|K|}$ one-sided convolution operators,
each involving $d-|K|$ active coordinate differences.
By the preceding theorem, each has MPO rank
at most $2^{d-|K|}r$. Since the $2^{d-|K|}$ of them are summed, each term in the sum has rank at most $2^{d-|K|}\times 2^{d-|K|}r = 4^{d-|K|}r$, with the exception of the term corresponding to $K=\{1,\ldots,d\}$, $f(\boldsymbol 0)I$,
which has rank one rather than $r$. Since there are $\binom dk$ subsets $K$ with $|K| = k$,
\[
\begin{aligned}
\operatorname{rank}_{\mathrm{MPO}}
(\widetilde F)
&\le
\sum_{k=0}^{d-1}
\binom dk 4^{d-k}r+1\\
&=(5^d-1)r+1.
\end{aligned}
\]

This covers the moment-to-local tensor of \Cref{sec:method}, which is a
two-sided convolution with the near-field stencil masked out; the mask is
absorbed into the kernel and does not increase the bound.
\end{remark}

\section{Algorithm: \texttt{DiffVortex}}\label{sec:method}

The procedure for one time step is summarized in \Cref{fig:pipeline}.

\begin{figure*}[!htbp]
\centering
\begin{tikzpicture}[
  font=\footnotesize,
  box/.style={rounded corners=3pt, align=center, text width=4.2cm, inner sep=5pt,
    execute at begin node={\hyphenpenalty=10000\exhyphenpenalty=10000}},
  phase/.style={box, draw=blue!55!cyan!75!black, fill=blue!6, minimum height=2.75cm},
  offline/.style={box, draw=orange!75!red!70!black, fill=orange!12, minimum height=1.3cm},
  arr/.style={-{Latex[length=2.2mm]}, thick},
  alab/.style={font=\scriptsize, inner sep=1.5pt, align=center},
]
\node[phase] (p1) at (0,0) {{\bfseries I.\ Velocity evaluation}\\[3pt]
  bin particles; moments $M_\alpha(\mathcal{C}_q)$\\
  contract with $\widehat F$ \eqref{eq:far_field}; evaluate\\ local expansions\\
  near-field direct sum \eqref{eq:biot_savart}\\ \mbox{over $\|p-q\|_\infty\leq w$}};
\node[phase] (p2) at (6.15,0) {{\bfseries II.\ Boundary coupling}\\[3pt]
  solve for $\gamma_k$: no-penetration\\ \eqref{eq:no_penetration} + circulation constraint\\
  add sheet-induced velocity\\ (fast or direct summation)};
\node[phase] (p3) at (12.3,0) {{\bfseries III.\ Evolution}\\[3pt]
  advect particles (RK2)\\ update body $\boldsymbol{y}_k$\\ (prescribed or elastic)\\
  inject wall vorticity; PSE\\ remesh if needed};
\node[offline] (off) at (0,2.8) {{\bfseries Offline, once}\\[2pt]
  QTCI compression of\\ $F_{p-q,\gamma,\alpha}$ \eqref{eq:far_field} $\to$ MPO $\widehat F$};
\draw[arr] (off) -- node[alab, right=2pt]{$\widehat F$, fixed} (p1);
\draw[arr] (p1) -- node[alab, above=1pt]{$\boldsymbol u_\omega$} (p2);
\draw[arr] (p2) -- node[alab, above=1pt]{$\boldsymbol u$} (p3);
\draw[arr] (p3.south) -- ++(0,-0.45) -| node[alab, pos=0.25, below=1pt]{$t\to t+\Delta t$:\; particles and body move, $\widehat F$ unchanged} (p1.south);
\coordinate (lb) at ($(p1.south)+(0,-0.75)$);
\begin{scope}[on background layer]
\node[draw=gray!60, dashed, rounded corners=5pt, inner sep=6pt, fit=(p1)(p2)(p3)(lb)] (loop) {};
\end{scope}
\node[anchor=south east, font=\footnotesize\itshape, text=gray!60!black] at (loop.north east) {every time step};
\end{tikzpicture}
%
\caption{One time step of \texttt{DiffVortex}. The inputs are the particle positions $\{\boldsymbol{x}_j\}$ and strengths $\{\Gamma_j\}$, the surface points $\{\boldsymbol{y}_k\}$ with prescribed wall velocity $\boldsymbol{u}_{\partial\mathcal{B}}$, and the MPO $\widehat F$. The moment-to-local operator $\widehat F$ (orange) is constructed once, offline, by quantics tensor cross interpolation of the tensor in \eqref{eq:far_field}, and is independent of the body's motion and shape. Phase~I evaluates the velocity by contracting $\widehat F$ against the box moments in a single bit-level pass, evaluating the resulting local expansions at the particles, and adding a direct sum over the near-field stencil. Phase~II solves for the bound-sheet strengths and adds the sheet-induced velocity, computed with the singular kernel $\boldsymbol{K}$. Phase~III advects the particles, updates the body, and releases wall vorticity into the flow. Everything inside the dashed loop is repeated every time step; with the second-order Runge--Kutta (RK2) integrator, Phases~I and~II are evaluated at each stage. The total cost per step is $\mathcal{O}(NP^2\log N)$.}\label{fig:pipeline}
\end{figure*}
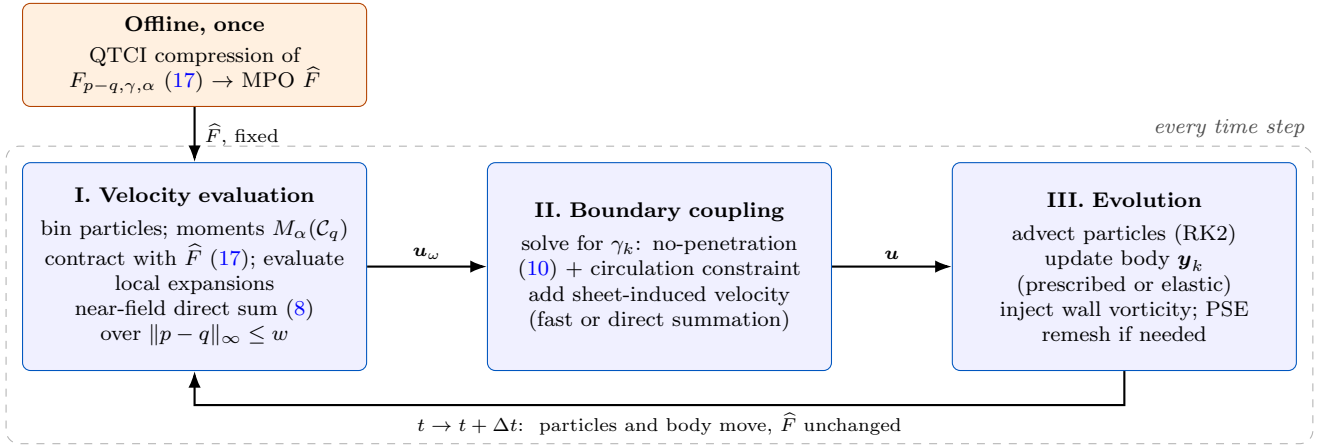

\subsection{Single-level fast multipole with MPO compression}
We cover the particle cloud with a uniform $\sqrt{b}\times\sqrt{b}$ grid of $b$ square boxes of side $h$, with centers $\boldsymbol{c}_p$. The compressed operator is built for a fixed choice of $x$ and $y$ limits, box side $h$ and blob radius $\delta$, so the grid is fixed for the whole run and is sized so that the wake does not reach its edge; the few particles that do leave are discarded, and the circulation they carry is monitored as a diagnostic. For two well-separated boxes, with center separation $R$ exceeding $(2+\sqrt{2})h$, the stream function $\psi$ of the blob-regularized Biot--Savart kernel admits a multipole expansion of order $P$, i.e.\ truncated after $P$ terms of a logarithmic series (\Cref{app:coeffs}),
\begin{equation}
\psi(\boldsymbol{t},\boldsymbol{s}) = \sum_{\gamma}\boldsymbol{t}^{\gamma}\sum_{\alpha} C_{\gamma,\alpha}(\boldsymbol{R})\,\boldsymbol{s}^{\alpha},
\end{equation}
where $\boldsymbol{t}$ and $\boldsymbol{s}$ are target and source positions relative to their box centers, $\boldsymbol{R}$ is the center-to-center displacement of the boxes, the multi-indices satisfy $|\gamma|+|\alpha|\le 2P$, and $C_{\gamma,\alpha}$ are explicit coefficients derived in \Cref{app:coeffs}. The induced velocity in a target box is then a sum over source-box moments $M_\alpha(\mathcal{C}_q) = \sum_{j\in \mathcal{C}_q}\Gamma_j\,\boldsymbol{s}_j^{\alpha}$, with $\boldsymbol{s}_j=\boldsymbol{x}_j-\boldsymbol{c}_q$.

The far-field contribution to the velocity at a particle $\boldsymbol{x}_i$ in box $p$, with $\boldsymbol{t}_i=\boldsymbol{x}_i-\boldsymbol{c}_p$, is
\begin{equation}
\begin{aligned}
u^{\mathrm{far}}_i &= -\sum_{|\gamma|\ge 1}\gamma_2\,\boldsymbol{t}_i^{\gamma-(0,1)}\sum_{\alpha}\sum_{q} F_{p-q,\gamma,\alpha}\,M_\alpha(\mathcal{C}_q), \\
v^{\mathrm{far}}_i &= \sum_{|\gamma|\ge 1}\gamma_1\,\boldsymbol{t}_i^{\gamma-(1,0)}\sum_{\alpha}\sum_{q} F_{p-q,\gamma,\alpha}\,M_\alpha(\mathcal{C}_q),
\end{aligned}
\label{eq:far_field}
\end{equation}
with $F_{p-q,\gamma,\alpha} = \mathbbm{1}_{\{\|p-q\|_\infty>w\}}\,C_{\gamma,\alpha}(\boldsymbol{c}_p-\boldsymbol{c}_q)$,
where $\mathbbm{1}_{\{\|p-q\|_\infty>w\}}$ excludes the near-field stencil, a $(2w+1)\times(2w+1)$ block of boxes of half-width $w$. This half-width is dictated by convergence of the multipole series: the expansion converges for all particle positions provided that $R > (2+\sqrt{2})h \approx 3.41h$. The nearest box pair on the ring $\|p-q\|_\infty = m$ is the axis-aligned one, at separation $mh$, so convergence requires the first far ring to satisfy $m \ge 4$, i.e.\ $w \ge 3$. We take $w = 4$ throughout. The near-field cost is $\mathcal{O}(N)$ at fixed box occupancy $N/b$, with $(2w+1)^2N/b = 81\,N/b$ pair interactions per particle; since the far-field contraction is instead $\mathcal{O}(bP^2)$ and independent of $N$, the two terms trade off against each other and $b$ should be chosen to grow with $N$ (see \Cref{sec:results-time}). For each fixed $(\gamma,\alpha)$, $F_{p-q,\gamma,\alpha}$ is a translation-invariant function of $p-q$ sampled on a uniform $\sqrt{b}\times\sqrt{b}$ grid, essentially the setting of~\Cref{thm:lifting}; for any fixed $p,q$, $F_{p - q,\gamma,\alpha}$ has the binomial (Pascal-type) structure of \Cref{app:coeffs} in $(\gamma,\alpha)$, which in practice also compresses well (\Cref{sec:results-rank}). Indexing the tensor with the interleaved bits of $(p,q)$ followed by the interleaved bits of $(\gamma,\alpha)$, with the multi-indices packed triangularly ($|\gamma|+|\alpha|\le2P$), and applying quantics tensor cross interpolation~\cite{nunez2025learning}, we obtain an MPO representation $\widehat F$ whose ranks do not grow with $b$ for the kernels studied here (\Cref{sec:results-rank}), consistent with smoothness-based rank bounds~\cite{lindsey2023multiscale}. The contraction $\sum_q F_{p-q,\gamma,\alpha} M_\alpha(\mathcal{C}_q)$ is then performed in $\mathcal{O}(bP^2\chi^2)$ operations up to logarithmic factors, where $\chi$ is the MPO bond dimension. Since $\chi$ does not grow with $b$ for the kernels studied here (\Cref{sec:results-rank}), the factor $\chi^2$ is a bounded constant and the cost is $\mathcal{O}(bP^2)$ in $b$ and $P$.

\subsection{Complexity}
For uniformly distributed particles with $b\sim N$, the cost per velocity evaluation breaks down as
\begin{equation}
\begin{aligned}
T(N) &= \underbrace{\mathcal{O}(NP^2)}_{\text{moments}} + \underbrace{\mathcal{O}\bigl((2w+1)^2 N^2/b\bigr)}_{\text{near field}} \\
&+ \underbrace{\mathcal{O}(bP^2 \log b)}_{\text{far-field MPO contraction}}
\quad+ \underbrace{\mathcal{O}(NP^2)}_{\text{local expansion eval.}} \\
&= \mathcal{O}(NP^2\log N)
\end{aligned}
\end{equation}
i.e.\ linear in the number of particles, up to a logarithmic factor, for fixed $P$, $w$ and occupancy $N/b$. Each RK2 step requires two such evaluations, and the construction of $\widehat F$ is a one-time offline cost (\Cref{app:compression}). The boundary coupling adds the cost of the no-penetration solve, $\mathcal{O}(N_\partial^3)$ when solved directly, together with the evaluation of the sheet-induced velocity at the particles. The latter is performed as a direct $\mathcal{O}(N N_\partial)$ sum over panels, using the unregularized kernel with an analytic near-panel correction; with $N_\partial = \mathcal{O}(\sqrt{N})$ under uniform refinement this is $\mathcal{O}(N^{3/2})$ and is the subdominant term only while $N_\partial$ remains small compared with the bulk cloud. The classical $b^2 P^4$ barrier of single-level FMM~\cite{martinsson2019fast} is removed by the MPO contraction, whose bounded rank is explained by the lifting theorem.

\subsection{Implementation}
\texttt{DiffVortex} is implemented in PyTorch. All operator structure (the box geometry, the bit-interleaved index pattern of $\widehat F$, the contraction order, and the near-field stencils) is fixed when the operator is built, so the per-step computation graph is a sequence of dense, differentiable tensor contractions over a fixed index structure. Within the fast summation, the particle-to-moment accumulation, the near-field blob kernel and the local-expansion evaluation are compiled and fused with \texttt{torch.compile}~\cite{ansel2024pytorch2} in its shape-polymorphic mode, so the changing particle and per-box counts do not force recompilation at every step. The moment-to-local contraction itself is left in eager mode. Since every stage is expressed in standard tensor algebra, the full time-stepping loop is compatible with automatic differentiation: gradients of a downstream objective (drag, thrust, or a shape functional) with respect to particle strengths, body kinematics, or geometric parameters can be obtained directly by backpropagating through the simulation, without hand-derived adjoint equations. A few steps are not smooth in the parameters: box assignment, the remeshing lattice and the number of released particles are piecewise constant, so the computed gradient is that of the locally smooth branch. \Cref{fig:gradient-validation} shows that it agrees with centered finite differences to ten significant figures. Reverse mode through the MPO contraction requires some care: the interleaved layout that bounds the forward sweep's working set also holds the number of live intermediates near its maximum for most of the sweep, so naive reverse-mode accumulation retains all $\mathcal{O}(\log b)$ of them at once. We therefore checkpoint each sweep segment, which bounds the backward peak to one or two intermediates in exchange for recomputing the contraction once. This also sidesteps the dynamic tree traversal that complicates classical FMM implementations on accelerator hardware, and the solver maps directly onto GPUs (\Cref{sec:results-time}).

\section{Numerical results}\label{sec:examples}
We organize the numerical results into three groups: (i) MPO rank diagnostics for the compressed M2L operator, which determine the feasibility of the framework; (ii) compute-time scaling for randomly distributed particles, which tests the predicted near-linear, $\mathcal{O}(N\log N)$, complexity; and (iii) physical benchmarks, namely an impulsively started cylinder, inviscid elliptical vortices, and moving, multiple and deforming bodies, together with a check of gradients through the full time integration.

\subsection{MPO rank diagnostics for the M2L operator}\label{sec:results-rank}

\Cref{fig:rank-bsqrt,fig:rank-sensitivity} report the MPO ranks of the moment-to-local tensor $F_{p-q,\gamma,\alpha}$ constructed by quantics tensor cross interpolation, as a function of the grid size $\sqrt{b}$, the regularization parameter $\delta$, and the QTCI relative tolerance. The numerical experiments use the two-dimensional blob-regularized Biot--Savart kernel of~\eqref{eq:biot_savart_kernel} and the multipole expansion of~\eqref{eq:far_field}, on the unit square, with the near-field mask $\mathbbm{1}_{\{\|p-q\|_\infty>4\}}$ of \Cref{sec:method} baked into the compressed tensor. The sweep is a full product grid: for each pair $(P,\delta)$ with $P\in\{3,4\}$ and $\delta\in\{10^{-1},10^{-2},10^{-3},10^{-4},10^{-5}\}$, the grid size is swept over $\sqrt{b}\in\{256,512,1024,2048,4096,8192\}$. The QTCI relative tolerance is $10^{-8}$ in this sweep and is varied separately in \Cref{fig:rank-sensitivity}(b). The bit layout is the interleaved, least-significant-bit-first ordering with triangular packing of the multipole indices (\Cref{sec:bitinterleaved,sec:method}).

Cross-interpolation libraries typically report the error at the most recently selected pivot, but we find that using this pivot error as a convergence criterion leads to premature termination. We therefore determine all reported ranks using a held-out criterion: the tensor train is evaluated on a random sample of index patterns not used as pivots, and the sweep terminates once this validation error reaches the prescribed tolerance. A subsequent TT-rounding sweep then truncates the QTT to the same tolerance~\cite{nunez2025learning}.

\begin{figure*}[!htbp]
\centering
\includegraphics[width=1\linewidth]{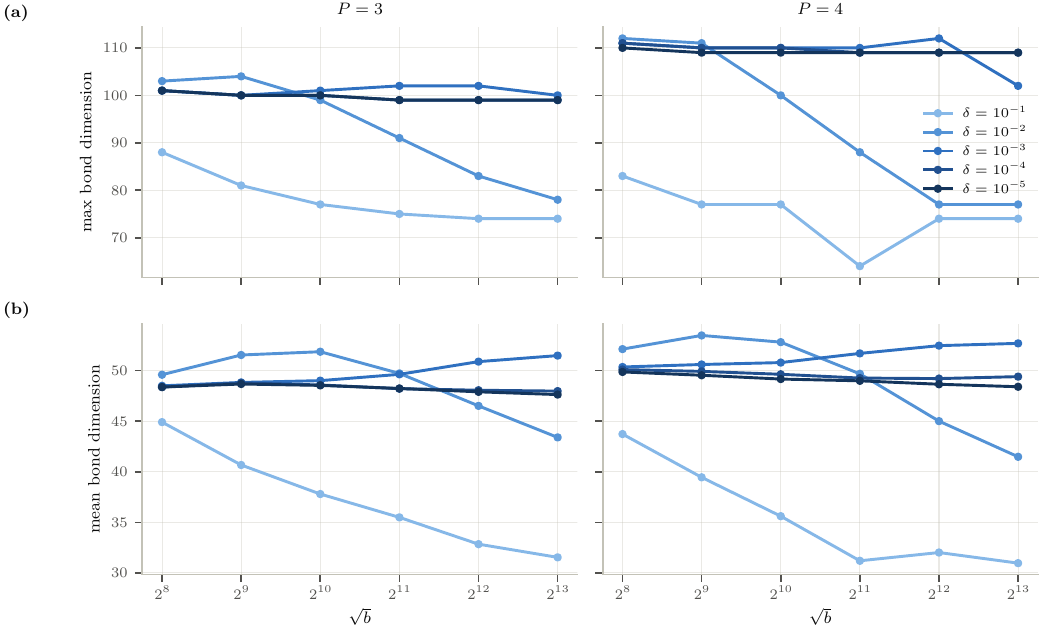}
\caption{MPO bond dimension of the moment-to-local tensor $F_{p-q,\gamma,\alpha}$ against the number of boxes per side $\sqrt{b}$, at QTCI relative tolerance $10^{-8}$. Columns: $P=3$ (left) and $P=4$ (right); one curve per regularization parameter $\delta\in\{10^{-1},\dots,10^{-5}\}$.
\textbf{(a)}~Maximum bond dimension. It does not grow with $\sqrt{b}$ over a $32$-fold range: for $\delta\le10^{-3}$ it is essentially constant (about $100$ for $P=3$ and $110$ for $P=4$), and for $\delta=10^{-1}$ and $10^{-2}$ it decreases as the grid is refined. This is consistent with \Cref{thm:lifting} combined with smoothness-based rank bounds for the kernel~\cite{lindsey2023multiscale}.
\textbf{(b)}~Mean bond dimension over the internal bonds, roughly half the maximum, so the maximum overstates the storage cost of the MPO.}\label{fig:rank-bsqrt}
\end{figure*}

\begin{figure*}[!htbp]
\centering
\includegraphics[width=1\linewidth]{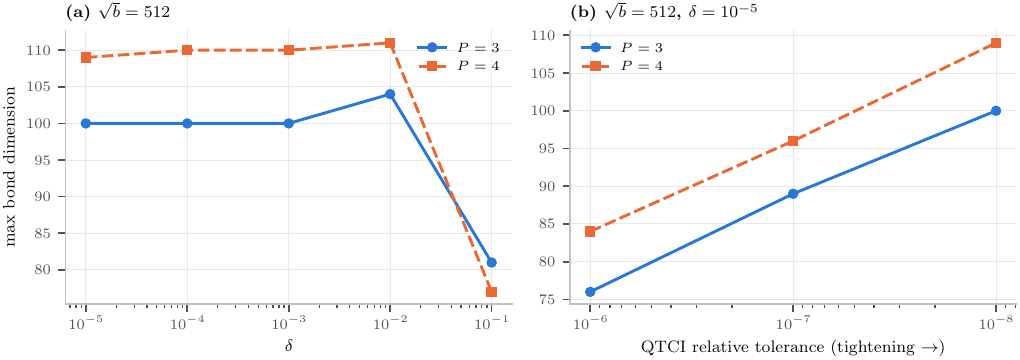}
\caption{Sensitivity of the maximum MPO bond dimension to the regularization parameter and the QTCI tolerance, at $\sqrt{b}=512$, for $P=3$ and $P=4$.
\textbf{(a)}~Against $\delta$, at tolerance $10^{-8}$ (a fixed-$\sqrt{b}$ slice of the sweep in \Cref{fig:rank-bsqrt}). The bond dimension rises between $\delta=10^{-1}$ and $10^{-2}$ and is essentially constant for $\delta\le10^{-3}$.
\textbf{(b)}~Against the QTCI relative tolerance, at $\delta=10^{-5}$. Over the three tolerances tested ($10^{-6}$, $10^{-7}$, $10^{-8}$) the bond dimension grows by about $12$ per decade of tolerance, consistent with logarithmic growth in $1/\text{tolerance}$ over this range.}\label{fig:rank-sensitivity}
\end{figure*}

These diagnostics are consistent with~\Cref{thm:lifting}: the MPO bond dimension of the M2L operator is essentially independent of the grid size $b$ (\Cref{fig:rank-bsqrt}) and increases and then saturates as the kernel's singularity is approached, and grows slowly as the QTCI tolerance is tightened (\Cref{fig:rank-sensitivity}). The construction is therefore robust for the blob-regularized kernels used throughout the rest of this paper.

\subsection{Compute-time scaling}\label{sec:results-time}

\Cref{tab:timing} and \Cref{fig:timing} report the wall-clock time per velocity evaluation for uniformly distributed random particles on the unit square. Two regimes are reported, because no single experiment covers the whole range: direct $\mathcal{O}(N^2)$ summation is exact and is the correct baseline, but it is only affordable up to $N\sim10^6$, whereas the compressed solver is run well past that.

The two panels of \Cref{fig:timing} hold different quantities fixed. In~(a) the box count $b$ is held \emph{fixed} while $N$ varies. This is the regime traversed within a single run, in which $b$ is necessarily fixed: the compressed operator $\widehat F$ is built once, offline, and reused as the particle count grows over the course of a simulation. It also separates the two terms of the complexity estimate cleanly. The far-field contraction costs $\mathcal{O}(bP^2)$ and is independent of $N$, whereas the near field costs $(2w+1)^2\,N\,(N/b)$---linear in $N$ at fixed occupancy, but growing as the occupancy $N/b$ rises. A run at fixed $b$ therefore traverses a curve on which the far-field cost is flat and the near-field cost grows, and the balance between the two is what determines how $b$ should be sized for a given $N$. Over the range where direct summation remains affordable, the compressed solver's cost is dominated by the terms that do not grow with $N$: the $\mathcal{O}(bP^2)$ contraction, and a near-field sweep that visits every occupied box whether that box holds two particles or two hundred. At $\sqrt{b}=64$ and $N=10^4$ there are only ${\sim}2.4$ particles per box, and the measured time is dominated by the near-field sweep, essentially a per-box overhead that saturates once most boxes are occupied and that accounts for the bulk of the measured time at this occupancy, the cost being dominated by visiting every target box and launching the corresponding GPU kernels rather than by arithmetic. The effect is visible in the two curves: at $\sqrt{b}=64$, where the boxes are already populated at $N=10^4$, the time rises only from $1.06$ to $1.17$\,s as $N$ triples, whereas at $\sqrt{b}=128$, where they are not, it rises from $2.13$ to $3.94$\,s over the same interval as the boxes fill. The curve is correspondingly flat, and the crossover $N^*$ that (a) reports is the point at which the $\mathcal{O}(N^2)$ baseline overtakes that fixed cost.

Panel~(b) holds the occupancy $N/b$ fixed instead, refining the box grid as $\sqrt{b}\propto\sqrt{N}$. Along that line the near-field work per particle is constant by construction and the far-field cost grows in proportion to $b$, so the total is linear in $N$. This is also the sizing rule a practitioner should follow: choose $b$ so that the occupancy stays in the range where the per-box work amortizes the per-box overhead, and the method's cost then tracks $N$ over more than two decades on CPU; on GPU the measured slope rises to ${\approx}1.5$ above $N\approx2\times10^6$. \Cref{app:compression} reports the one-time cost of constructing the compressed operator.
\begin{table}[!htbp]
\centering
\caption{Wall-clock time per velocity evaluation for $N$ uniformly distributed random particles on the unit square, at fixed $\sqrt{b}$. All runs use $P=4$, $\delta=10^{-3}$, near-field stencil $\|p-q\|_\infty\le4$ and QTCI tolerance $10^{-7}$, on a single A100 GPU (MIG 3g.40gb partition). \emph{Fast} is the MPO-compressed solver (compressed far field plus direct near field); \emph{Direct} is exact $\mathcal{O}(N^2)$ pairwise summation on the same device, so each row compares algorithms rather than hardware. Speed-up is direct time divided by fast time. The last column is $10^4$ times the relative $L^2$ error of the fast velocity against the direct result, and is therefore an end-to-end error for the whole pipeline rather than a property of the compressed operator alone. Compression is a one-off cost per $(P,\sqrt{b},\delta)$, amortized over every evaluation in a run (\Cref{app:compression}).}\label{tab:timing}
\small
\setlength{\tabcolsep}{4pt}
\input{figs/timing_table.tex}
\end{table}

\begin{figure*}[!htbp]
\centering
\includegraphics[width=\linewidth]{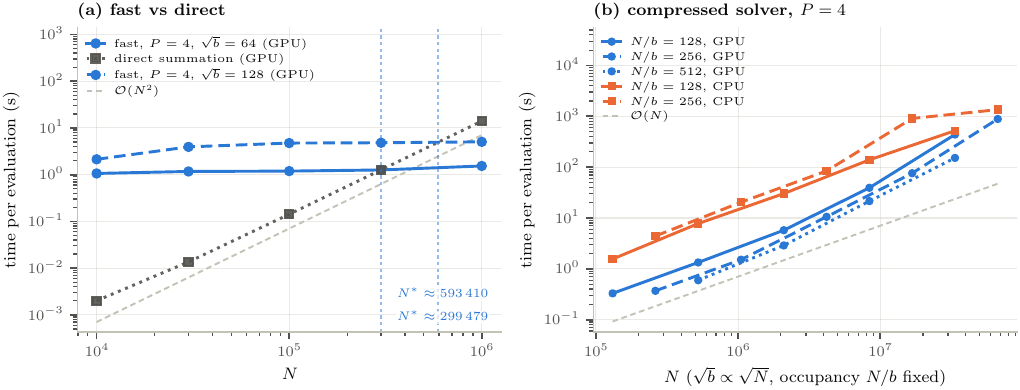}
\caption{Compute-time scaling of \texttt{DiffVortex} for $N$ particles drawn uniformly on the unit square, at $P=4$, $\delta=10^{-3}$, near-field stencil $\|p-q\|_\infty\le4$ and QTCI tolerance $10^{-7}$. GPU timings use a single A100 (MIG 3g.40gb partition); CPU timings use eight cores. ``Fast'' denotes the MPO-compressed solver.
\textbf{(a)}~Fixed $\sqrt{b}$, $N$ from $10^{4}$ to $10^{6}$: the compressed solver against exact $\mathcal{O}(N^{2})$ direct summation on the same GPU, which is affordable only up to $N=10^{6}$ ($10^{12}$ pairwise evaluations) and supplies the end-to-end error in \Cref{tab:timing}. Dotted vertical lines mark the crossover $N^{*}$ beyond which the compressed solver is faster: $N^{*}\approx3.0\times10^{5}$ at $\sqrt{b}=64$ and $N^{*}\approx5.9\times10^{5}$ at $\sqrt{b}=128$.
\textbf{(b)}~Fixed occupancy $N/b$, so that $\sqrt{b}$ grows as $\sqrt{N}$. GPU runs extend to $N=3.4\times10^{7}$ at occupancies $128$ and $512$ and to $N=6.7\times10^{7}$ at occupancy $256$; CPU runs, at occupancies $128$ and $256$ only, extend to $3.4\times10^{7}$ and $6.7\times10^{7}$, respectively. Direct summation is unaffordable throughout this range.}\label{fig:timing}
\end{figure*}

We note that the compressed operator's own held-out validation error (i.e., the quantity the QTCI construction converges against) is several orders of magnitude smaller than the end-to-end error of the resulting velocity field. The two measure different things. The former estimates how well the tensor train reproduces the moment-to-local tensor $F_{p-q,\gamma,\alpha}$; the latter also contains the truncation error of the multipole expansion at order $P$, which for the well-separated pairs retained here scales as $|z|^{P}$ with $|z|\le(2\sqrt{2}hR+2h^2)/\Lambda^2$ (\Cref{app:coeffs}); at the closest retained separation $R=(w+1)h=5h$ this bound is $0.65$, and it is attained only by particles in the far corners of the two boxes, so the typical error is far smaller. At the orders used in practice it is the multipole truncation, not the tensor compression, that sets the accuracy of the method: tightening the QTCI tolerance below that level buys nothing but rank. The tolerance should accordingly be chosen just fine enough to remain invisible against the $P$-truncation error. At fixed $\sqrt{b}$, both the near-field stencil half-width $w$ and the multipole order $P$ can be raised if more accuracy is required, at the cost of a new offline construction of $\widehat F$.

\subsection{Physical benchmarks}\label{sec:results-physical}

\paragraph{Impulsively started cylinder.}
We first validate the viscous formulation against the high-resolution vortex-method calculations of Koumoutsakos and Leonard~\cite{koumoutsakos1995high} for an impulsively started circular cylinder. Starting from an initially irrotational fluid, vorticity is generated at the cylinder surface and released into the fluid through the wall vorticity flux procedure described in \Cref{sec:formulation}. Following~\cite{koumoutsakos1995high}, we define $\mathrm{Re}=UD/\nu$ and the nondimensional time $T=2Ut/D$, where $U$ is the free-stream speed and $D$ the cylinder diameter. \Cref{fig:cylinder-drag}(a) compares the transient drag over $0\leq T\leq 0.5$ at $\mathrm{Re}=40$, $200$, $1000$, and $10^4$ with the reference calculations of~\cite{koumoutsakos1995high}. This interval captures the initial development of the flow following the impulsive start, during which thin boundary layers form along the cylinder and separate into vortex sheets that become thinner as the Reynolds number increases. The same precomputed MPO is used throughout each calculation. The regularization $\delta$ may be chosen independently of the compression parameters; here it is tied to the boundary-layer thickness and so varies with Reynolds number. The ranks depend only weakly on it, varying by a factor of $1.4$ across four decades of $\delta$ (\Cref{fig:rank-bsqrt}), so the compressed interaction costs essentially the same at every Reynolds number considered. \Cref{fig:cylinder-drag}(b) shows the drag at $\mathrm{Re} = 9500$ over $0\leq T \leq 5$. The results are in good agreement over $0 \leq T \leq 4$, after which the two curves diverge, our drag remaining near $C_D\approx1.2$ where the reference falls to $\approx 0.9$, possibly due to the chaotic evolution of vortices at high $\mathrm{Re}$~\cite{aref1983integrable}.

\begin{figure*}[!htbp]
\centering
\includegraphics[width=1\linewidth]{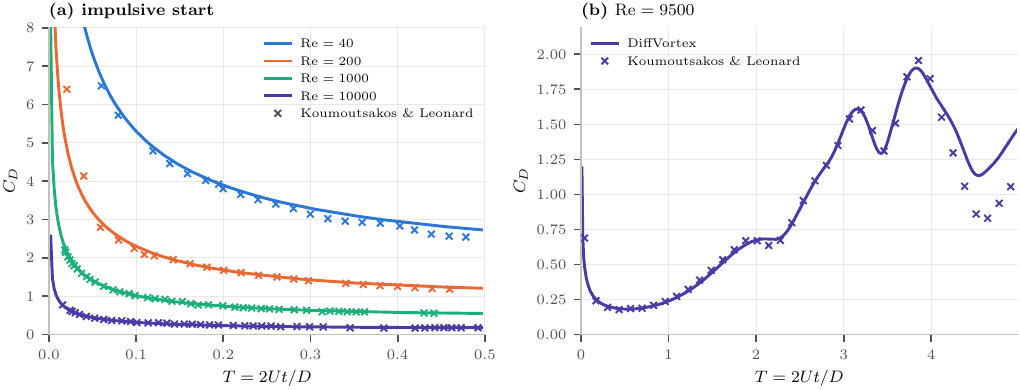}
\caption{Transient drag coefficient $C_D$ on an impulsively started circular cylinder, against the nondimensional time $T=2Ut/D$, where $U$ is the free-stream speed and $D$ the cylinder diameter. Lines: \texttt{DiffVortex}; crosses: the vortex-method calculations of Koumoutsakos and Leonard~\cite{koumoutsakos1995high}.
\textbf{(a)}~Early-time drag, $0\le T\le 0.5$, at $\mathrm{Re}=40$, $200$, $1000$ and $10^4$; the two calculations agree closely at all four Reynolds numbers.
\textbf{(b)}~$\mathrm{Re}=9500$ over $0\le T\le 5$; the curves agree closely up to $T\approx4$ and diverge thereafter (see text).}
\label{fig:cylinder-drag}
\end{figure*}

\Cref{fig:cylinder-high-re} additionally compares the vorticity field at $\mathrm{Re}=9500$ and $T=5$ with the corresponding late-time vortex dynamics reported
in~\cite{koumoutsakos1995high}. By this time the initially shed layers have rolled up into compact vortices connected to the cylinder by thin sheets, and the two calculations agree well.

\begin{figure*}[!htbp]
\centering
\includegraphics[width=\linewidth]{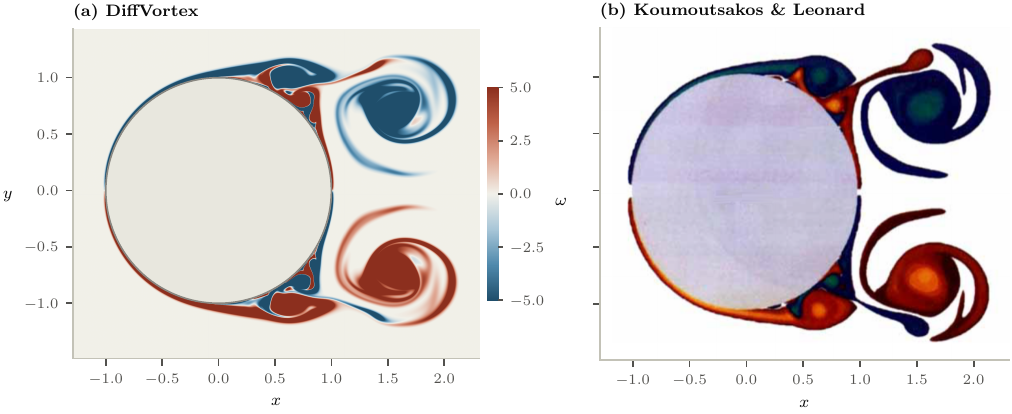}
\caption{Vorticity field $\omega$ of the impulsively started cylinder at $\mathrm{Re}=9500$ and $T=5$.
\textbf{(a)}~\texttt{DiffVortex}.
\textbf{(b)}~The corresponding calculation of Koumoutsakos and Leonard~\cite{koumoutsakos1995high}, reproduced from their Fig.~26 with permission.}
\label{fig:cylinder-high-re}
\end{figure*}

\paragraph{Inviscid elliptical vortices.}
To isolate the accuracy of the particle interaction from viscous
vorticity production at immersed boundaries, we next consider freely
evolving vortices in the inviscid limit. As an exact benchmark, we use a
Kirchhoff vortex: a patch of uniform vorticity with elliptical support. With this initial condition, the incompressible Euler equations admit an exact solution in which the elliptical vorticity patch rotates uniformly, like a rigid body, preserving its elliptical shape~\cite{lamb1932hydrodynamics,xu2023dynamics}. In \Cref{fig:elliptical-vortices} we compare the measured orientation of the numerical vortex with the exact solution; the agreement is excellent. The aspect ratio is likewise conserved, to within a relative error of order $10^{-4}$ and with no secular growth over the full revolution. We then consider the continuously distributed Gaussian elliptical vortices studied by Xu and Krasny~\cite{xu2023dynamics}. Unlike the Kirchhoff patch, these vortices deform internally and generate
filamentary structures as they evolve, providing a nontrivial test of the
long-range interaction in a setting where no immersed-boundary treatment
is involved. The dynamics of the Gaussian vortex agree closely with those reported by Xu and Krasny~\cite{xu2023dynamics}, whose Fig.~7 (column~1) shows the same vortex at the same five times.

\begin{figure*}[!htbp]
\centering
\includegraphics[width=\linewidth]{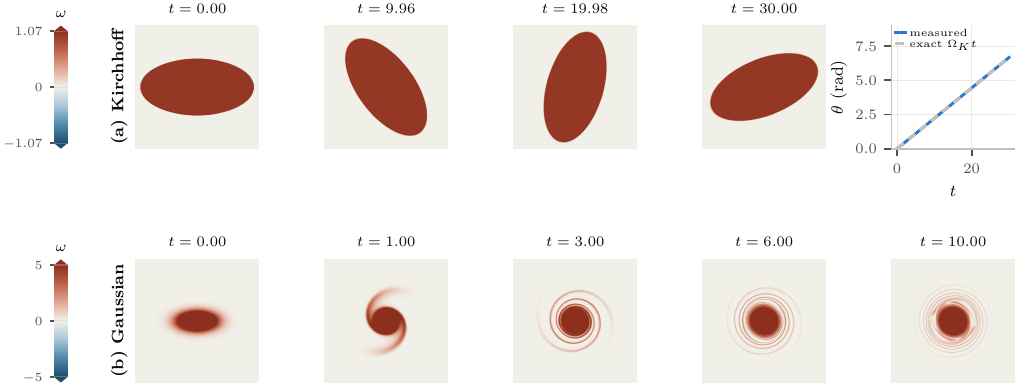}
\caption{Inviscid evolution of elliptical vortices; color shows the vorticity $\omega$.
\textbf{(a)}~Kirchhoff elliptical vortex patch at $t=0$, $9.96$, $19.98$ and $30$, and its measured orientation $\theta(t)$ against the exact uniform rotation $\theta=\Omega_K t$, where $\Omega_K=\omega_0ab/(a+b)^2$ for a patch of vorticity $\omega_0$ with semi-axes $a$ and $b$. The patch rotates rigidly and preserves its shape, as the exact solution predicts.
\textbf{(b)}~Gaussian elliptical vortex of Xu and Krasny~\cite{xu2023dynamics} at $t=0$, $1$, $3$, $6$ and $10$. The vortex axisymmetrizes while shedding thin spiral filaments, a strongly deforming test of the long-range interaction. These panels correspond to Fig.~7 (column~1) of Xu and Krasny~\cite{xu2023dynamics}, which shows a wider range of vorticity amplitudes.}
\label{fig:elliptical-vortices}
\end{figure*}

\paragraph{Multiple moving and deforming bodies.}
We next demonstrate the geometric flexibility gained by retaining a
Lagrangian description of the immersed boundaries. In the first example,
three circular cylinders are placed in tandem and undergo prescribed transverse
oscillations with alternating phase. Each cylinder generates its own
surface vorticity while interacting with the wakes of the other bodies.
In the second example, an ellipse in a uniform stream
undergoes a prescribed periodic deformation in which its transverse
semi-axis varies sinusoidally in time, the streamwise semi-axis being slaved to it so that the enclosed area is exactly conserved. The boundary points simply follow
the prescribed evolving ellipse, while the same particle and
boundary-vorticity machinery is used at every configuration. These
examples are intended as demonstrations of moving, multibody, and
deforming-boundary functionality rather than as separate physical
validation cases (\Cref{fig:moving-bodies}). In both cases $\mathrm{Re} = 200$, based on the rest diameter $D=1$ and free-stream speed $U=1$, with $\Delta t = 1.63\times10^{-2}$, $368$ surface points per body, $\delta = \Delta x = 8.54\times10^{-3}$, and a compressed operator at $P=4$, $\sqrt{b}=128$ and near-field stencil $w=4$. The three cylinders are spaced two diameters apart and oscillate transversely with amplitude $0.25D$ at unit frequency, in alternating phase; the ellipse deforms at frequency $0.5$ with semi-axis amplitude $0.25$, sweeping aspect ratios from $4\!:\!1$ to $1\!:\!2.25$.

\begin{figure*}[!htbp]
\centering
\includegraphics[width=\linewidth]{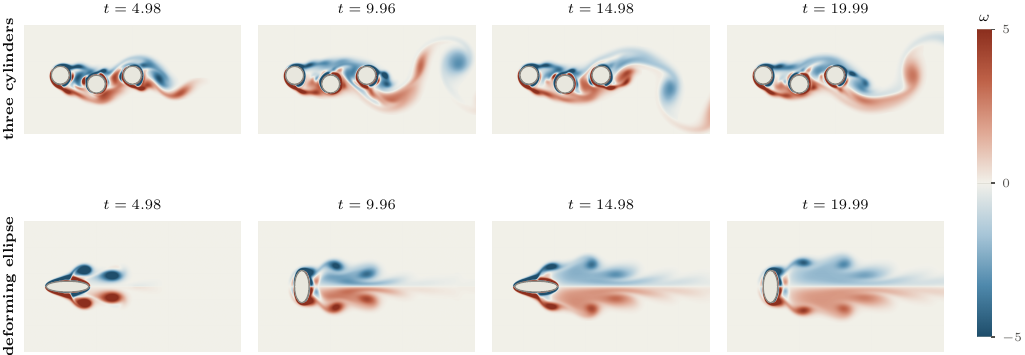}
\caption{Vorticity $\omega$ for flows with time-dependent immersed geometry at $\mathrm{Re}=200$, at $t\approx5$, $10$, $15$ and $20$.
Top: three cylinders in tandem undergoing prescribed transverse oscillations
with alternating phase.
Bottom: a periodically deforming ellipse with constant area in a uniform stream.
In both cases the body geometry changes while the precomputed MPO $\widehat F$ used for the long-range particle interaction remains unchanged.}
\label{fig:moving-bodies}
\end{figure*}

\paragraph{Differentiation through the flow solver.}
Finally, we verify that gradients can be propagated through the complete \texttt{DiffVortex} time integration. For the deforming ellipse, let $\epsilon$
denote the amplitude of the prescribed shape deformation and define the
finite-time drag objective
\begin{equation}
    J(\epsilon)
    =
    \int_{\Delta t}^{t_f} C_D(t;\epsilon)\,dt,
\end{equation}
where $C_D$ is the drag coefficient and $t_f$ the final time; the lower limit is $\Delta t$ because the impulse derivative supplying $C_D$ is a centered difference, which leaves no sample at the first or last step. The derivative is evaluated at $\epsilon=0.1$, with $\mathrm{Re}=100$, $\Delta t = 1.78\times10^{-2}$ and a $50$-step rollout, so $t_f \approx 0.870$.
Automatic differentiation gives $dJ/d\epsilon$ directly through the
particle trajectories, surface-vorticity generation, and compressed
far-field evaluation. We validate this derivative against centered finite
differences over a sequence of perturbation sizes (\Cref{fig:gradient-validation}). Agreement over the
range bounded by round-off errors below provides an end-to-end check
that differentiation through the tensor-compressed particle solver
recovers the sensitivity of a physical observable to the prescribed body
motion.

\begin{figure}[!htbp]
\centering
\includegraphics[width=0.9\linewidth]{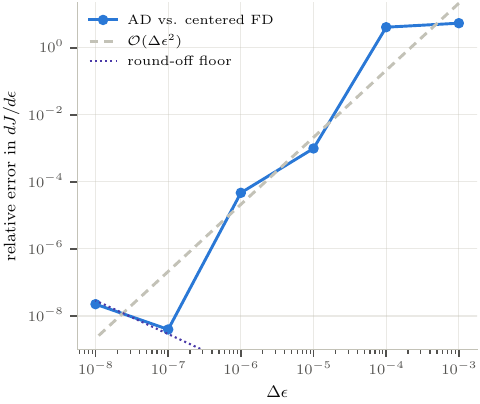}
\caption{End-to-end differentiation through \texttt{DiffVortex}, with the compressed far-field operator engaged at every step. Relative difference between the automatic-differentiation derivative $dJ/d\epsilon$ of the integrated drag and centered finite differences with step $\Delta\epsilon$ (solid), against a least-squares $\mathcal{O}(\Delta\epsilon^{2})$ guide line (dashed) and the centered-difference round-off floor $\varepsilon_{\mathrm{mach}}|J|/(\Delta\epsilon\,|J'|)$ (dotted). The finite differences converge to the automatic-differentiation result at a fitted order of $2.01$ between $\Delta\epsilon=10^{-3}$ and $10^{-8}$. This convergence is non-monotonic since the remeshing lattice makes the particle count piecewise constant in $\epsilon$. The floor rises as $\Delta\epsilon$ shrinks, since differencing ever closer values of $J$ loses absolute precision. The two smallest steps lie on it, at $1.4$ and $0.8$ times its value. Hence the rise from $\Delta\epsilon=10^{-7}$ to $10^{-8}$ can be explained by round-off errors.}
\label{fig:gradient-validation}
\end{figure}
\section{Discussion}\label{sec:discussion}

We have presented a quantum-inspired solver for two-dimensional incompressible flows past moving and deforming bodies. Its starting point is that, in Lagrangian vortex methods, the expensive long-range operation is convolution with a fixed, smooth, translation-invariant kernel, so it is the operator, not the field, that should be compressed. In interleaved binary coordinates the one-sided convolution operator inherits a low MPO rank from the MPS rank of the kernel itself, with a sharp constant of $2^d$, and once it is precomputed as an MPO, immersed bodies enter only through their boundary integral: flexible or moving bodies require no additional compression machinery.

The framework complements rather than competes with the Eulerian MPS solvers of~\cite{peddinti2024quantum,holscher2025quantum,vanhulst2025curvilinear,pisoni2025compression}. Eulerian MPS solvers are at their best when the velocity field itself admits a low-rank MPS description, a regime documented mainly for two-dimensional decaying turbulence and moderate-Reynolds-number flows in simple geometries. The present Lagrangian framework is designed for regimes in which the flow develops thin filamentary structures, when the body deforms in time, or when multiple bodies move independently. In such cases we expect the operator-compression view to be the more natural one, although a head-to-head comparison with Eulerian MPS solvers remains to be done. A unified solver that switches between the two compressions adaptively, depending on the local flow regime, is an attractive direction for future work.

Several extensions are natural. The extension to three dimensions requires the three-dimensional vector Biot--Savart kernel together with the vortex-stretching term; \Cref{thm:lifting} then applies with $d=3$, and the worst-case MPO rank inflates by a factor $2^3=8$ relative to the kernel rank, which itself may be larger in three dimensions. Fluid--structure interaction proper, in which the body responds to the fluid forces, requires coupling the body update in Phase~III of \Cref{fig:pipeline} to an elastic body model. Probability-distribution-level descriptions of turbulence, recently computed with MPS by Gourianov \emph{et al.}~\cite{gourianov2025tensor}, are a further direction in which operator and state compression might be combined.

The approach also has limitations. The box grid is uniform, and the moment-to-local contraction runs over all $b$ boxes whether or not they hold particles, so its $\mathcal{O}(bP^{2})$ cost is paid even where the flow is empty (\Cref{sec:results-time}). The optimal scaling also requires the particles to be uniformly distributed. However, re-meshing is applied frequently to control Lagrangian distortion~\cite{koumoutsakos2008flow}, leaving the particles approximately uniformly distributed within the occupied boxes. The compressed operator depends on $\delta$, $P$, the box grid and the near-field stencil, so changing any of these requires a new offline construction (\Cref{app:compression}). Finally, all results presented here are two-dimensional.

More broadly, the same operator-compression strategy should carry over to other systems governed by smooth, translation-invariant kernels, such as regularized Stokes flow, electrostatics and gravitation.

\appendix
\section{Proof of the lifting theorem}\label{app:proof}

\begin{proof}[Proof of \Cref{thm:lifting}]
We prove both estimates using the finite-state representation of binary subtraction. The one-dimensional construction is common to the serial and interleaved orderings; their different rank bounds arise only when the construction is extended to $d$ dimensions.

We begin with $d=1$. Consider the sampled vector
\[
f_i=f\left(\frac{i}{2^L-1}\right).
\]
Throughout this appendix, we order the tensor cores so that the last core carries the least significant bit:
\[
i=\sum_{k=0}^{L-1}i^{(k)}2^{L-1-k}.
\]
This reverses the convention of \Cref{sec:bitinterleaved} but does not change the maximal MPS rank.

Suppose that $f_i$ has an MPS representation
\begin{align}
f_i
&=
G^{(0)}_{r_1}(i^{(0)})
G^{(1)}_{r_1r_2}(i^{(1)})
\nonumber\\
&\qquad
\cdots
G^{(L-2)}_{r_{L-2}r_{L-1}}(i^{(L-2)})
G^{(L-1)}_{r_{L-1}}(i^{(L-1)}),
\label{eq:kernel-mps}
\end{align}
where all bond dimensions are at most $r$. Naively replacing every bit $i^{(k)}$ by $i^{(k)}-j^{(k)}$ fails because the $k$th bit of $i-j$ is not, in general, the difference of the corresponding bits of $i$ and $j$. Working from the least significant bit to the most significant bit, one instead has
\begin{align*}
(i-j)^{(L-1)}
&=i^{(L-1)}-j^{(L-1)}\pmod 2,\\
(i-j)^{(L-2)}
&=i^{(L-2)}-j^{(L-2)}-\kappa_{L-1}\pmod 2,\\
&\hspace{1cm}\vdots\\
(i-j)^{(1)}
&=i^{(1)}-j^{(1)}-\kappa_2\pmod 2,\\
(i-j)^{(0)}
&=i^{(0)}-j^{(0)}-\kappa_1\pmod 2,
\end{align*}
where $\kappa_k\in\{0,1\}$ is the borrow generated at level $k$:
\[
\kappa_k=
\begin{cases}
0, & \text{if subtraction at level $k$ requires no borrow},\\
1, & \text{otherwise}.
\end{cases}
\]

The borrow can be included in the MPO bond index. Moving from the least to the most significant bit, each core receives the incoming borrow, evaluates the corresponding bit of $i-j$, and passes the outgoing borrow to the next core. More precisely, define
\begin{align}
&H^{(L-1)}_{r_{L-1},\kappa_{L-1}}
   \bigl(i^{(L-1)},j^{(L-1)}\bigr)
\nonumber\\
&\qquad =
G^{(L-1)}_{r_{L-1}}
\bigl(i^{(L-1)}-j^{(L-1)}\bmod 2\bigr)
\nonumber\\
&\qquad\quad\times
\mathbbm{1}_{
\left\{
\kappa_{L-1}
=
\mathbbm{1}_{\,i^{(L-1)}-j^{(L-1)}<0}
\right\}},
\label{eq:last-borrow-core}
\end{align}
and, for $k=1,\ldots,L-2$,
\begin{align}
&H^{(k)}_{r_k,\kappa_k,r_{k+1},\kappa_{k+1}}
   \bigl(i^{(k)},j^{(k)}\bigr)
\nonumber\\
&\qquad =
G^{(k)}_{r_kr_{k+1}}
\bigl(i^{(k)}-j^{(k)}-\kappa_{k+1}\bmod 2\bigr)
\nonumber\\
&\qquad\quad\times
\mathbbm{1}_{
\left\{
\kappa_k
=
\mathbbm{1}_{\,i^{(k)}-j^{(k)}-\kappa_{k+1}<0}
\right\}}.
\label{eq:interior-borrow-core}
\end{align}
At the most significant bit, define
\begin{align}
&H^{(0)}_{r_1,\kappa_1}
   \bigl(i^{(0)},j^{(0)}\bigr)
\nonumber\\
&\qquad =
G^{(0)}_{r_1}
\bigl(i^{(0)}-j^{(0)}-\kappa_1\bmod 2\bigr)
\mathbbm{1}_{\{i^{(0)}-j^{(0)}-\kappa_1\geq0\}}.
\label{eq:first-borrow-core}
\end{align}
The indicator functions in \eqref{eq:last-borrow-core} and \eqref{eq:interior-borrow-core} propagate the unique admissible borrow sequence. The final indicator in \eqref{eq:first-borrow-core} eliminates precisely those pairs for which $i<j$. Consequently,
\begin{align}
(T_f)_{i,j}
&=
H^{(0)}_{r_1,\kappa_1}\bigl(i^{(0)},j^{(0)}\bigr)
H^{(1)}_{r_1,\kappa_1,r_2,\kappa_2}\bigl(i^{(1)},j^{(1)}\bigr)
\cdots
\nonumber\\
&\qquad\times
H^{(L-2)}_{r_{L-2},\kappa_{L-2},r_{L-1},\kappa_{L-1}}
\bigl(i^{(L-2)},j^{(L-2)}\bigr)
\nonumber\\
&\qquad\times
H^{(L-1)}_{r_{L-1},\kappa_{L-1}}
\bigl(i^{(L-1)},j^{(L-1)}\bigr)
\nonumber\\
&=f_{i-j}\mathbbm{1}_{i\geq j}.
\label{eq:lifted-mpo}
\end{align}
Every internal bond is indexed by $(r_k,\kappa_k)$ with $\kappa_k\in\{0,1\}$, so the maximal MPO rank is at most $2r$.

In the serial ordering, the bits belonging to each physical coordinate form a consecutive block. The preceding construction may therefore be applied independently within each coordinate block. At a bond inside a block, the kernel bond index $r_k$ is augmented by the single borrow bit for that coordinate. At the most significant bit of each block, the final indicator enforces $i_\mu\geq j_\mu$, and the borrow state is terminated before the next coordinate is processed. Thus at most one borrow bit is carried across any bond, giving

\[
R_{\mathrm S}\leq 2r_{\mathrm S}.
\]

This gives the same factor of two as the multilevel Toeplitz construction of~\cite[Theorem~4.6]{kazeev2013multilevel}, but here the rank $r_S$ refers to the QTT representation of the kernel sampled on $[0,1]^d$, rather than that of the full signed generating vector.

In the interleaved ordering, by contrast, the bits of all $d$ coordinates at a given dyadic level belong to the same core. Binary subtraction must therefore propagate the joint borrow vector

\[
\boldsymbol{\kappa}_k
=
(\kappa_{k,1},\ldots,\kappa_{k,d})
\in\{0,1\}^d
\]

between consecutive levels. Since this vector has $2^d$ possible states, enlarging each kernel bond index $r_k$ to $(r_k,\boldsymbol{\kappa}_k)$ gives

\[
R_{\mathrm I}\leq2^d r_{\mathrm I}.
\]

This proves both estimates.
\end{proof}

\begin{remark}
Although the construction above was presented for linear convolution, the same rank bounds apply to periodic convolution: subtraction is then taken modulo $2^L$, and the final indicator in \eqref{eq:first-borrow-core} is simply dropped. The lifting factor is independent of the grid resolution but depends on the chosen ordering and, through the kernel rank, on the physical dimension. For $d\leq3$, the explicit interleaved factor satisfies $2^d\leq8$.
\end{remark}

\subsection{Sharpness of the lifting constants}

We first establish sharpness in one dimension. Let $f(x)=a^x$ with $a>0$ and define
\[
q=a^{1/(2^L-1)}.
\]
The sampled kernel factorizes exactly across its binary modes:
\[
f\left(\frac{i}{2^L-1}\right)
=q^i
=\prod_{k=0}^{L-1}\left(q^{2^{L-1-k}}\right)^{i^{(k)}}.
\]
It therefore has exact MPS rank one~\cite{khoromskij2011quantics}. The corresponding one-sided (triangular) convolution operator is
\[
A_{i,i'}=q^{i-i'}\mathbbm{1}_{i\geq i'},
\qquad
i,i'\in\{0,\ldots,2^L-1\}.
\]

We show that every nontrivial unfolding of the quantized operator has rank two. Fix $k\in\{1,\ldots,L-1\}$ and partition the bits into the low sets
\[
\mathcal{I}_{\mathrm{lo}}=\{i^{(L-k)},\ldots,i^{(L-1)}\},
\qquad
\mathcal{I}'_{\mathrm{lo}}=\{i'^{(L-k)},\ldots,i'^{(L-1)}\},
\]
and the high sets
\[
\mathcal{I}_{\mathrm{hi}}=\{i^{(0)},\ldots,i^{(L-k-1)}\},
\qquad
\mathcal{I}'_{\mathrm{hi}}=\{i'^{(0)},\ldots,i'^{(L-k-1)}\}.
\]
Let $I_{\mathrm{lo}},I'_{\mathrm{lo}}\in\{0,\ldots,2^k-1\}$ and $I_{\mathrm{hi}},I'_{\mathrm{hi}}\in\{0,\ldots,2^{L-k}-1\}$ be the integers represented by these bit sets. Then
\[
i=I_{\mathrm{lo}}+2^kI_{\mathrm{hi}},
\qquad
i'=I'_{\mathrm{lo}}+2^kI'_{\mathrm{hi}}.
\]
The exponent decomposes as
\[
i-i'=I_{\mathrm{lo}}-I'_{\mathrm{lo}}+2^k(I_{\mathrm{hi}}-I'_{\mathrm{hi}}),
\]
while the triangular indicator decomposes lexicographically:
\[
\mathbbm{1}_{i\geq i'}
=
\mathbbm{1}_{I_{\mathrm{hi}}>I'_{\mathrm{hi}}}
+
\mathbbm{1}_{I_{\mathrm{hi}}=I'_{\mathrm{hi}}}
\mathbbm{1}_{I_{\mathrm{lo}}\geq I'_{\mathrm{lo}}}.
\]
It follows that
\begin{align}
A_{i,i'}
&=
\underbrace{q^{I_{\mathrm{lo}}-I'_{\mathrm{lo}}}}
_{u_1(I_{\mathrm{lo}},I'_{\mathrm{lo}})}
\underbrace{
q^{2^k(I_{\mathrm{hi}}-I'_{\mathrm{hi}})}
\mathbbm{1}_{I_{\mathrm{hi}}>I'_{\mathrm{hi}}}
}_{v_1(I_{\mathrm{hi}},I'_{\mathrm{hi}})}
\nonumber\\
&\quad+
\underbrace{
q^{I_{\mathrm{lo}}-I'_{\mathrm{lo}}}
\mathbbm{1}_{I_{\mathrm{lo}}\geq I'_{\mathrm{lo}}}
}_{u_2(I_{\mathrm{lo}},I'_{\mathrm{lo}})}
\underbrace{
\mathbbm{1}_{I_{\mathrm{hi}}=I'_{\mathrm{hi}}}
}_{v_2(I_{\mathrm{hi}},I'_{\mathrm{hi}})}.
\label{eq:exponential-unfolding}
\end{align}
Hence the unfolding matrix that separates the high bits from the $k$ low bits can be written as
\[
\mathcal{A}^{(k)}
=
\boldsymbol{u}_1\boldsymbol{v}_1^T
+
\boldsymbol{u}_2\boldsymbol{v}_2^T.
\]
The vectors $\boldsymbol{v}_1$ and $\boldsymbol{v}_2$ have disjoint nonempty supports and are therefore linearly independent. Likewise, $\boldsymbol{u}_1$ and $\boldsymbol{u}_2$ are linearly independent because $\boldsymbol{u}_1$ is nonzero everywhere, whereas $\boldsymbol{u}_2$ vanishes when $I_{\mathrm{lo}}<I'_{\mathrm{lo}}$. Thus
\[
\operatorname{rank}\mathcal{A}^{(k)}=2
\]
for every nontrivial unfolding, and the exact one-dimensional MPO rank is two.

For general $d$, consider the separable exponential
\[
f(x_1,\ldots,x_d)=\prod_{\mu=1}^d a_\mu^{x_\mu}.
\]
Its sampled tensor has exact MPS rank one in both the serial and interleaved orderings. The associated convolution operator is the Kronecker product
\[
T_f=A_1\otimes\cdots\otimes A_d,
\]
where each $A_\mu$ is a one-dimensional one-sided convolution operator of exact MPO rank two.

In the serial ordering, the binary modes of the factors $A_\mu$ occur in consecutive coordinate blocks. The ranks within each block are two, while the bonds separating consecutive coordinate blocks have rank one. Hence
\[
R_{\mathrm S}=2.
\]
In the interleaved ordering, the $k$th unfolding factorizes across all $d$ coordinates:
\[
\mathcal{T}^{(k)}
=
\mathcal{A}_1^{(k)}\otimes\cdots\otimes\mathcal{A}_d^{(k)}.
\]
Therefore,
\[
\operatorname{rank}\mathcal{T}^{(k)}
=
\prod_{\mu=1}^d\operatorname{rank}\mathcal{A}_\mu^{(k)}
=2^d,
\]
and hence
\[
R_{\mathrm I}=2^d.
\]
The constants $2$ and $2^d$ in \Cref{thm:lifting} are therefore both sharp.

\section{Numerical comparison of the two orderings}
\label{app:ordering-comparison}

The sharp constants describe exact rank inflation relative to the kernel rank in the same ordering. They do not compare the kernel ranks $r_{\mathrm S}$ and $r_{\mathrm I}$ themselves, since these ranks arise from different tensor unfoldings. In particular, bounded rank in one ordering does not imply bounded rank in the other.

Here we compare the numerical MPO ranks produced by the two orderings at a fixed compression tolerance. \Cref{tab:ordering-ranks} reports the ranks of the M2L tensor $F_{p-q,\gamma,\alpha}$ on $\sqrt{b}\times \sqrt{b}$ grids using least-significant-bit-first ordering and a prescribed relative tolerance of $10^{-8}$. The average is taken over the internal MPO bonds.

\begin{table}[!htbp]
\centering
\caption{Maximum and average MPO ranks of the moment-to-local tensor $F_{p-q,\gamma,\alpha}$ on $\sqrt{b}\times\sqrt{b}$ box grids, for the interleaved and serial orderings, at QTCI relative tolerance $10^{-8}$. Both orderings place the least significant bit first within each coordinate and use triangular packing of the multipole indices. All runs use $P=3$, $\delta=10^{-2}$ and near-field stencil $\|p-q\|_\infty\le4$; the average is over the internal bonds.}
\label{tab:ordering-ranks}
\small
\setlength{\tabcolsep}{3pt}
\begin{tabular}{ccrrr}
\toprule
$\sqrt{b}$ &
Ordering &
Max.\ rank &
Avg.\ rank &
Parameters \\
\midrule
16  & Interleaved & 100 & 31.60 & 86\,320\\
16  & Serial      & 100 & 26.92 & 62\,288\\
\addlinespace
32  & Interleaved & 101 & 39.59 & 150\,502\\
32  & Serial      & 116 & 33.41 & 117\,274\\
\addlinespace
64  & Interleaved & 103 & 44.52 & 202\,796\\
64  & Serial      & 129 & 39.27 & 183\,856\\
\addlinespace
128 & Interleaved & 103 & 47.30 & 243\,266\\
128 & Serial      & 140 & 45.19 & 273\,122\\
\bottomrule
\end{tabular}
\end{table}

The maximum interleaved rank remains nearly constant over the tested range, increasing only from $100$ to $103$ as $\sqrt{b}$ increases from $16$ to $128$. By contrast, the maximum serial rank increases from $100$ to $140$. Thus, despite its larger formal lifting factor, the interleaved ordering exhibits substantially better maximum-rank stability for this kernel.

The average ranks give a more qualified comparison. The serial ordering has the smaller average rank at every grid size tested, although the difference is smallest at the largest grid. At $\sqrt{b}=128$, the average ranks are $47.30$ and $45.19$ in the interleaved and serial orderings, respectively. The parameter counts are not governed by the average rank alone, since they scale with the squares of the bond dimensions and the serial rank profile is the more peaked. As the table shows, the serial representation is the smaller at $\sqrt{b}\le64$, by $28\%$ at $\sqrt{b}=16$, but the ordering reverses at $\sqrt{b}=128$, where it exceeds the interleaved count by $12\%$.

These results illustrate why the exact lifting constants alone do not determine which ordering is preferable. Serial ordering carries only one borrow bit at a time but contains tensor cuts separating complete physical coordinates. The ranks across these cuts measure global coordinate-separation complexity and may grow with problem size. Interleaved ordering instead groups the coordinates at each dyadic scale. For the kernel considered here, the resulting maximum numerical rank remains nearly bounded even though its universal exact lifting factor is larger.

The distinction between exact and approximate rank is also essential. \Cref{thm:lifting} constructs exact representations and does not characterize the singular-value decay of their tensor unfoldings. The ranks in \Cref{tab:ordering-ranks} are instead approximate ranks obtained after truncation at a prescribed tolerance and may be substantially smaller than the exact construction bounds.

\section{Derivation of the moment-to-local coefficients}\label{app:coeffs}

We derive the two-box moment-to-local operator for the blob-regularized Biot--Savart kernel~\eqref{eq:biot_savart_kernel}, i.e.\ the explicit coefficients $C_{\gamma,\alpha}(\boldsymbol{R})$ appearing in \Cref{sec:method}. Suppose the targets $\boldsymbol{x}_i$ lie in a box $\mathcal{C}_t$ with center $\boldsymbol{c}_t$ and side length $h$, and the sources $\boldsymbol{x}'_j$ lie in a box $\mathcal{C}_s$ with center $\boldsymbol{c}_s$ and side length $h$, with $\mathcal{C}_s$ and $\mathcal{C}_t$ well separated so that, taking $\boldsymbol{R} = \boldsymbol{c}_t- \boldsymbol{c}_s$, we have $R:= |\boldsymbol{R}| > (2+\sqrt{2})h$. Let $\Lambda^2 := R^2+ \delta^2$ and define the box-local coordinates $\boldsymbol{t} = \boldsymbol{x} - \boldsymbol{c}_t$ and $\boldsymbol{s} = \boldsymbol{x}' - \boldsymbol{c}_s$. Define
\begin{align*}
\psi(\boldsymbol{t},\boldsymbol{s}) &= \frac{1}{4\pi}\log(|\boldsymbol{x} - \boldsymbol{x}'|^2 + \delta^2) - \frac{1}{2\pi}\log\Lambda \\
&= \frac{1}{4\pi}\log(|\boldsymbol{t} - \boldsymbol{s} + \boldsymbol{R}|^2 + \delta^2) - \frac{1}{2\pi}\log\Lambda,
\end{align*}
so that the induced velocity components are $u = -\partial_{t_2}\psi$ and $v = \partial_{t_1}\psi$. Expanding,
\begin{align*}
    \psi(\boldsymbol{t},\boldsymbol{s}) &= \frac{1}{4\pi}\log(|\boldsymbol{t} - \boldsymbol{s}|^2 + 2(\boldsymbol{t} - \boldsymbol{s})\cdot\boldsymbol{R} + R^2 + \delta^2)\\
    &\quad- \frac{1}{2\pi}\log\Lambda\\
    &= \frac{1}{4\pi}\log\Big(1 + \frac{|\boldsymbol{t} - \boldsymbol{s}|^2 + 2(\boldsymbol{t} - \boldsymbol{s})\cdot\boldsymbol{R}}{\Lambda^2}\Big)\\
    &= \frac{1}{4\pi}\sum^\infty_{n =1} (-1)^{n+1}\frac{z^n}{n},
\end{align*}
where $z = \big(|\boldsymbol{t} - \boldsymbol{s}|^2 + 2(\boldsymbol{t} - \boldsymbol{s})\cdot\boldsymbol{R}\big)/\Lambda^2$. To see that the well-separation condition ensures convergence, note that $z = \big(|\boldsymbol{x}-\boldsymbol{x}'|^2 + \delta^2 - \Lambda^2\big)/\Lambda^2$, so $|z|<1$ if and only if $|\boldsymbol{x}-\boldsymbol{x}'|^2 + \delta^2 < 2\Lambda^2$. Since $|\boldsymbol{t}|,|\boldsymbol{s}| \le h/\sqrt{2}$, we have $|\boldsymbol{x}-\boldsymbol{x}'| \le R + \sqrt{2}h$, and the series therefore converges for \emph{all} particle positions in the two boxes provided $(R+\sqrt{2}h)^2 < 2R^2$, i.e.\ $R > \sqrt{2}h/(\sqrt{2}-1) = (2+\sqrt{2})h \approx 3.41h$, the well-separation condition stated above. On the uniform grid of \Cref{sec:method}, this is guaranteed by assigning the boxes with $\|p-q\|_\infty \le 3$ to the near field; the implementation uses the slightly more conservative $\|p-q\|_\infty \le 4$ (see \Cref{sec:method}). Let $\boldsymbol{\xi} = \boldsymbol{t}-\boldsymbol{s}$. Then
\begin{align*}
z^n &= \frac{1}{\Lambda^{2n}}\big(|\boldsymbol{\xi}|^2 + 2(\boldsymbol{\xi}\cdot\boldsymbol{R})\big)^n\\
&= \frac{1}{\Lambda^{2n}}\sum_{m=0}^n \binom{n}{m}(|\boldsymbol{\xi}|^2)^m(2\boldsymbol{\xi}\cdot\boldsymbol{R})^{n - m},
\end{align*}
with
\begin{align*}
(|\boldsymbol{\xi}|^2)^m &= \sum_{r = 0}^m \binom{m}{r}\xi_1^{2r}\xi_2^{2m-2r}, \\
(2\boldsymbol{\xi}\cdot\boldsymbol{R})^{n-m} &= 2^{n-m}\sum_{p = 0}^{n-m} \binom{n-m}{p}R_1^pR_2^{n - m -p}\\
&\qquad\times\xi_1^{p}\xi_2^{n - m -p}.
\end{align*}
Hence
\begin{align*}
z^n &= \frac{1}{\Lambda^{2n}}\sum_{m=0}^n\sum_{r = 0}^m\sum_{p = 0}^{n - m}2^{n-m}\binom{n}{m}\binom{m}{r}\binom{n-m}{p} \\
&\quad\times R_1^pR_2^{n - m -p}\xi_1^{p + 2r}\xi_2^{n + m -p -2r}.
\end{align*}
Letting $\beta = (\beta_1, \beta_2) = (p+2r,\ n+m-p-2r)$ be a multi-index, this becomes
\begin{equation*}
z^n = \frac{1}{\Lambda^{2n}} \sum_\beta A_{n,\beta}(\boldsymbol{R})\,\boldsymbol{\xi}^\beta,
\end{equation*}
\begin{align*}
A_{n,\beta}(\boldsymbol{R}) &=  \sum_{\substack{0 \leq m \leq n,\ 0\leq r \leq  m\\ 0 \leq p \leq n-m \\ p+2r = \beta_1 \\
n+m -p-2r = \beta_2}}2^{n-m}\binom{n}{m}\binom{m}{r}\binom{n-m}{p} \\
&\quad\times R_1^pR_2^{n - m -p},
\end{align*}
where by construction $|\beta| := \beta_1 + \beta_2 = n + m \leq 2n$, so the sum over $\beta$ is finite. Truncating at multipole order $P$,
\begin{equation*}
\psi(\boldsymbol{t},\boldsymbol{s}) \approx \frac{1}{4\pi}\sum_\beta B_{P,\beta}(\boldsymbol{R})\,\boldsymbol{\xi}^\beta,
\end{equation*}
\begin{equation*}
B_{P,\beta}(\boldsymbol{R}) = \sum_{n = 1}^P \frac{(-1)^{n+1}}{n\Lambda^{2n}}A_{n,\beta}(\boldsymbol{R}),
\end{equation*}
with $|\beta| \leq 2P$. By the binomial expansion,
\begin{equation*}
\boldsymbol{\xi}^\beta = (\boldsymbol{t} - \boldsymbol{s})^\beta = \sum_{\gamma \leq \beta}\binom{\beta}{\gamma}(-1)^{|\beta - \gamma|}\boldsymbol{t}^\gamma \boldsymbol{s}^{\beta - \gamma},
\end{equation*}
where $\binom{\beta}{\gamma} = \binom{\beta_1}{\gamma_1}\binom{\beta_2}{\gamma_2}$ and $|\gamma| \leq 2P$. Hence
\begin{align}
    \psi(\boldsymbol{t},\boldsymbol{s})&= \frac{1}{4\pi}\sum_\beta\sum_{\gamma \leq \beta} B_{P,\beta}(\boldsymbol{R})\binom{\beta}{\gamma}(-1)^{|\beta - \gamma|}\boldsymbol{t}^\gamma \boldsymbol{s}^{\beta - \gamma} \notag\\
    &= \sum_{\gamma}\boldsymbol{t}^\gamma \sum_\alpha C_{\gamma,\alpha}(\boldsymbol{R})\boldsymbol{s}^{\alpha},
\end{align}
where
\begin{align*}
C_{\gamma,\alpha}(\boldsymbol{R}) &= \frac{1}{4\pi}\binom{\alpha + \gamma}{\gamma}(-1)^{|\alpha|}\sum_{n = 1}^P \frac{(-1)^{n+1}}{n\Lambda^{2n}} \\
&\quad\times\sum_{ \substack{0 \leq m \leq n\\0\leq r \leq  m\\ 0 \leq p \leq n-m \\ p+2r = \alpha_1 + \gamma_1 \\
n+m -p-2r = \alpha_2 + \gamma_2}} 2^{n-m}\binom{n}{m}\binom{m}{r}\binom{n-m}{p} \\
&\quad\times R_1^pR_2^{n - m -p}.
\end{align*}

\section{Operator construction cost}\label{app:compression}

\Cref{tab:timing} reports the cost of \emph{applying} the compressed operator. The operator must first be built, and that cost is incurred once per $(P,\sqrt{b},\delta)$ and then amortized over every velocity evaluation of every simulation that reuses it: the compressed tensor train depends only on $P$, $\sqrt{b}$, $\delta$ and the near-field stencil. \Cref{tab:compression} records that cost, and the size of the resulting operator, including those used in \Cref{fig:timing}. The cost grows with both parameters, but sublinearly in $\sqrt{b}$: a sixteenfold increase in $\sqrt{b}$ raises it roughly fourfold.

The compressed operator is also small. Every entry in \Cref{tab:compression}
is roughly one to two megabytes, and the growth with $\sqrt{b}$ is weak: at $P=4$ a sixteenfold increase in $\sqrt{b}$ (a $256$-fold increase in the box count), from $\sqrt{b}=32$ to $512$, enlarges the operator by only $50\%$. This is the practical consequence of the rank behavior of \Cref{fig:rank-bsqrt}. The tensor train holds $\mathcal{O}(n_c\chi^{2})$ parameters for $n_c$ cores and maximal rank $\chi$; refining the grid adds cores only logarithmically in $b$, while $\chi$ stays flat. The ten operators in \Cref{tab:compression} occupy about $15$\,MB in total.

\begin{table}[htbp]
\centering
\caption{Cost of constructing the compressed moment-to-local operator $\widehat F$ by QTCI, and its size, for $P=3$ and $P=4$, at $\delta=10^{-3}$, near-field stencil $\|p-q\|_\infty\le4$ and QTCI tolerance $10^{-7}$ (including the operators used in \Cref{fig:timing}). Build times are wall-clock seconds on eight CPU cores. Size is the storage of the tensor-train cores: the parameter count times eight bytes (float64).}\label{tab:compression}
\input{figs/compression_table.tex}
\end{table}

\begin{acknowledgments}
SV acknowledges support from NSF under grant DMS-2513346.
This research was supported in part through computational resources and services provided by Advanced Research Computing at the University of Michigan, Ann Arbor.
\end{acknowledgments}
\clearpage

\section*{Code availability}
\texttt{DiffVortex} is implemented in PyTorch and will be released under an open-source license soon.

\bibliographystyle{apsrev4-2}
\bibliography{refs}

\end{document}

%% file: figs/timing_table.tex
\begin{tabular}{rrrrrr}
\toprule
$\sqrt{b}$ & $N$ & Fast (s) & Direct (s) & Speed-up & $10^4\times$ Error \\
\midrule
64 & $10^{4}$ & 1.063 & 0.002 & 0.0$\times$ & 6.81 \\
64 & $3\times 10^{4}$ & 1.172 & 0.014 & 0.0$\times$ & 7.32 \\
64 & $10^{5}$ & 1.193 & 0.142 & 0.1$\times$ & 6.54 \\
64 & $3\times 10^{5}$ & 1.262 & 1.266 & 1.0$\times$ & 6.72 \\
64 & $10^{6}$ & 1.536 & 14.025 & 9.1$\times$ & 6.66 \\
128 & $10^{4}$ & 2.133 & 0.002 & 0.0$\times$ & 7.19 \\
128 & $3\times 10^{4}$ & 3.935 & 0.014 & 0.0$\times$ & 7.58 \\
128 & $10^{5}$ & 4.748 & 0.143 & 0.0$\times$ & 6.74 \\
128 & $3\times 10^{5}$ & 4.824 & 1.268 & 0.3$\times$ & 7.09 \\
128 & $10^{6}$ & 5.048 & 14.032 & 2.8$\times$ & 6.99 \\
\bottomrule
\end{tabular}

%% file: figs/compression_table.tex
\begin{tabular}{rrrrr}
\toprule
& \multicolumn{2}{c}{Build time (s)} & \multicolumn{2}{c}{Size (MB)} \\
\cmidrule(lr){2-3} \cmidrule(lr){4-5}
$\sqrt{b}$ & $P=3$ & $P=4$ & $P=3$ & $P=4$ \\
\midrule
32 & 379 & 632 & 0.97 & 1.26 \\
64 & 484 & 836 & 1.26 & 1.52 \\
128 & 800 & 1\,137 & 1.47 & 1.68 \\
256 & 1\,280 & 1\,518 & 1.61 & 1.81 \\
512 & 1\,598 & 2\,564 & 1.73 & 1.91 \\
\bottomrule
\end{tabular}